\documentclass[12pt]{article}
\usepackage{geometry,cite}                
\usepackage[utf8]{inputenc}
\usepackage{amsmath, amsthm, amsfonts, amssymb,mathrsfs}
\usepackage{graphicx}
\usepackage{lmodern}
\usepackage{mathabx}
\usepackage{color}
\usepackage{esint}
\usepackage{caption,subcaption}

\newtheorem{theorem}{Theorem}[section]
\newtheorem{definition}[theorem]{Definition}
\newtheorem{proposition}[theorem]{Proposition}
\newtheorem{lemma}[theorem]{Lemma}
\newtheorem{remark}[theorem]{Remark}

\newcommand{\cL}{\mathcal{L}}
\newcommand{\cR}{\mathcal{R}}
\newcommand{\Z}{\mathbb{Z}}
\newcommand{\R}{\mathbb{R}}
\newcommand{\T}{\mathbb{T}}
\newcommand{\C}{\mathbb{C}}
\newcommand{\G}{\mathbb{G}}
\newcommand{\Prob}{\mathbb{P}}
\newcommand{\bx}{\mathbf{x}}
\newcommand{\by}{\mathbf{y}}
\newcommand{\ba}{\mathbf{a}}

\newcommand{\bc}{\mathbf{c}}
\newcommand{\bk}{\mathbf{k}}
\newcommand{\bn}{\mathbf{n}}
\newcommand{\bm}{\mathbf{m}}
\newcommand{\bp}{\mathbf{p}}
\newcommand{\bq}{\mathbf{q}}
\newcommand{\br}{\mathbf{r}}
\newcommand{\bA}{\mathbf{A}}

\newcommand{\bE}{\mathbf{E}}

\newcommand{\bgamma}{\boldsymbol{\gamma}}
\newcommand{\bzero}{\boldsymbol{0}}
\newcommand{\tGamma}{\vec{\Gamma}}
\newcommand{\tq}{\vec{\bq}}
\newcommand{\tp}{\vec{\bp}}
\newcommand{\starb}{*_{\mathscr B}}
\newcommand{\Trace}[1]{\mathcal{T}_\Prob \left ( #1 \right )}

\newcommand{\be}{\mathbf{e}}

\newcommand{\dps}{\displaystyle }

\newcommand{\class}[2]{\overline{#2}^{\cR_#1}}

\author{Eric Canc\`es \\ CERMICS, Ecole des Ponts and Inria Paris \\ 6 \& 8 avenue
Blaise Pascal\\77455 Marne-la-Vall\'ee Cedex 2, France\\ \texttt{eric.cances@enpc.fr}
\and
Paul Cazeaux \\ School of Mathematics\\
University of Minnesota\\
Minneapolis, Minnesota 55455\\
\texttt{pcazeaux@umn.edu}
\and Mitchell Luskin\\ School of Mathematics\\
University of Minnesota\\
Minneapolis, Minnesota 55455\\
\texttt{luskin@umn.edu}
}
\title{Generalized Kubo Formulas for the Transport Properties of Incommensurate 2D Atomic Heterostructures}

\begin{document}
\maketitle
\begin{abstract}
We give an exact formulation for the transport coefficients of incommensurate two-dimensional atomic multilayer systems in the tight-binding approximation. This formulation is based upon the $C^*$ algebra framework introduced
by Bellissard and collaborators~\cite{bellissard1994noncommutative,bellissard2003coherent} to study aperiodic solids (disordered crystals, quasicrystals, and amorphous materials), notably in the presence of magnetic fields (quantum Hall effect).
We also present numerical approximations and test our methods on a one-dimensional incommensurate bilayer system.
 \end{abstract}

\section{Introduction}
The synthesis and modeling of layered two-dimensional atomic heterostructures is currently being intensely investigated
with the goal of designing materials with desired electronic and optical properties~\cite{Geim2013}.  These multilayer two-dimensional materials
are generally incommensurate, that is, the multilayer system does not have a periodic structure although each individual layer
does have a periodic structure.

Despite the scientific and technological importance of
incommensurate materials, an exact
formulation has not yet been given for important properties such as the electronic density of
states or the electrical conductivity.  The typical approach to modeling and computing the properties of incommensurate structures is to approximate them by commensurate structures or supercells~\cite{Terrones2014}.  Although this approach might provide a good approximation in many cases, the error is generally uncontrolled (see \cite{cazeauxrippling} for an analysis of the supercell approximation of the mechanical relaxation of coupled incommensurate chains). Further, the approximation of the important case of small angle rotated bilayer structures requires supercell sizes that are too large for numerical solution~\cite{2DPerturb15}.  In this paper, we give such an exact
formulation and develop numerical approximations.

Within independent electron or mean field models such as Hartree-Fock or Kohn-Sham, the electronic density of states and transport properties for periodic structures can be rigorously formulated by the use of the Bloch transform to obtain generalized eigenstates of the Hamiltonian as a plane wave with wave vector in the Brillouin zone ($\bk$-space) multiplied by a periodic function.  This approach leads to the classical Kubo formulae for the transport properties of periodic solids which can be formulated as the trace of corresponding operators~\cite{Kaxiras_2003}.

Although incommensurate multilayer two-dimensional materials no longer have a periodic structure and the local environment of each atom is unique, this local environment can be simply characterized by shifts of each layer.  Further, the shifts corresponding to the local environments are uniformly distributed over the periodic unit cell of each layer.
One can then develop generalized Kubo formulae for incommensurate heterostructures by considering the integral over the uniformly distributed shifts of the trace of operators that depend of the corresponding local environment.

To give a precise formulation and enable mathematical and numerical analysis, we apply the $C^*$-algebra approach for aperiodic solids introduced
by Bellissard and collaborators~\cite{bellissard1994noncommutative,bellissard2003coherent} to incommensurate heterostructures. Following these lines, we can characterize the incommensurate structure by its hull, which is a compact description of the local environments. We first present the construction of the hull for perfect multilayers in Sections~\ref{subsec:geometry} and \ref{sec:hullcontinuous}.  The case of disordered heterostructures is dealt with in Section~\ref{subsec:disorder}.

In Section~\ref{sec:CSA}, we present the $C^*$-algebra formalism for tight-binding models.
Rather than working with operators on the infinite dimensional tight-binding state space, the $C^*$-algebra approach allows us to exploit the simplification of working directly with a $C^*$-algebra of functions which represent the operators of interest.
Within the tight-binding model that we consider, we obtain a compact parametrization of the tight-binding Hamiltonian by using environment-dependent site and hopping functions~\cite{Fang2015}.
The $C^*$-algebra approach then allows us to concisely construct transport operators and their trace by utilizing the algebra structure. After introducing the abstract setting in Section~\ref{subsec:abstractsetting}, we describe the special case of perfect incommensurate bilayers in full detail in Section~\ref{sec:bilayers}.

In Section~\ref{sec:numerics}, we present a new, minimalistic one-dimensional toy model to showcase the expected effects of incommensurability in coupled multilayered systems.
We chose to introduce here a simple discretization based on periodic supercells and on the Kernel Polynomial Method~\cite{Weisse2006}. 
In future works, we will develop more sophisticated approaches based on the $C^*$-algebras introduced in this paper. We will present their numerical analysis,  using the $C^*$-algebra formalism, and use the minimalistic model introduced in this paper as a benchmark to analyze and test different implementation strategies, possibly targeting directly incommensurate cases, for computing the density of states, the conductivity or other observables. One such possible strategy is to exploit locality~\cite{MassattDOS16,Carr2016}.

Previous research on the development of numerical methods to approximate transport properties within the $C^*$-algebra formulation has been done by Prodan~\cite{prodan2012quantum} for the effects of disorder and magnetic fields.
Recent work on the analysis and computation of the density of states for incommensurate layers from the operator point of view is given in~\cite{MassattDOS16}. We note that the density of states is defined in~\cite{MassattDOS16} as a thermodynamic limit, while the density of states and transport properties are given explicit expressions in the $C^*$-algebra approach, which directly provide the values of the quantities of interest in the thermodynamic limit.

\section{Geometries of multi-layered systems}\label{sec:hull}

\subsection{Perfect multilayer structure}\label{subsec:geometry}

\begin{figure}[t!]
	\centering
	\includegraphics[width=.8\textwidth]{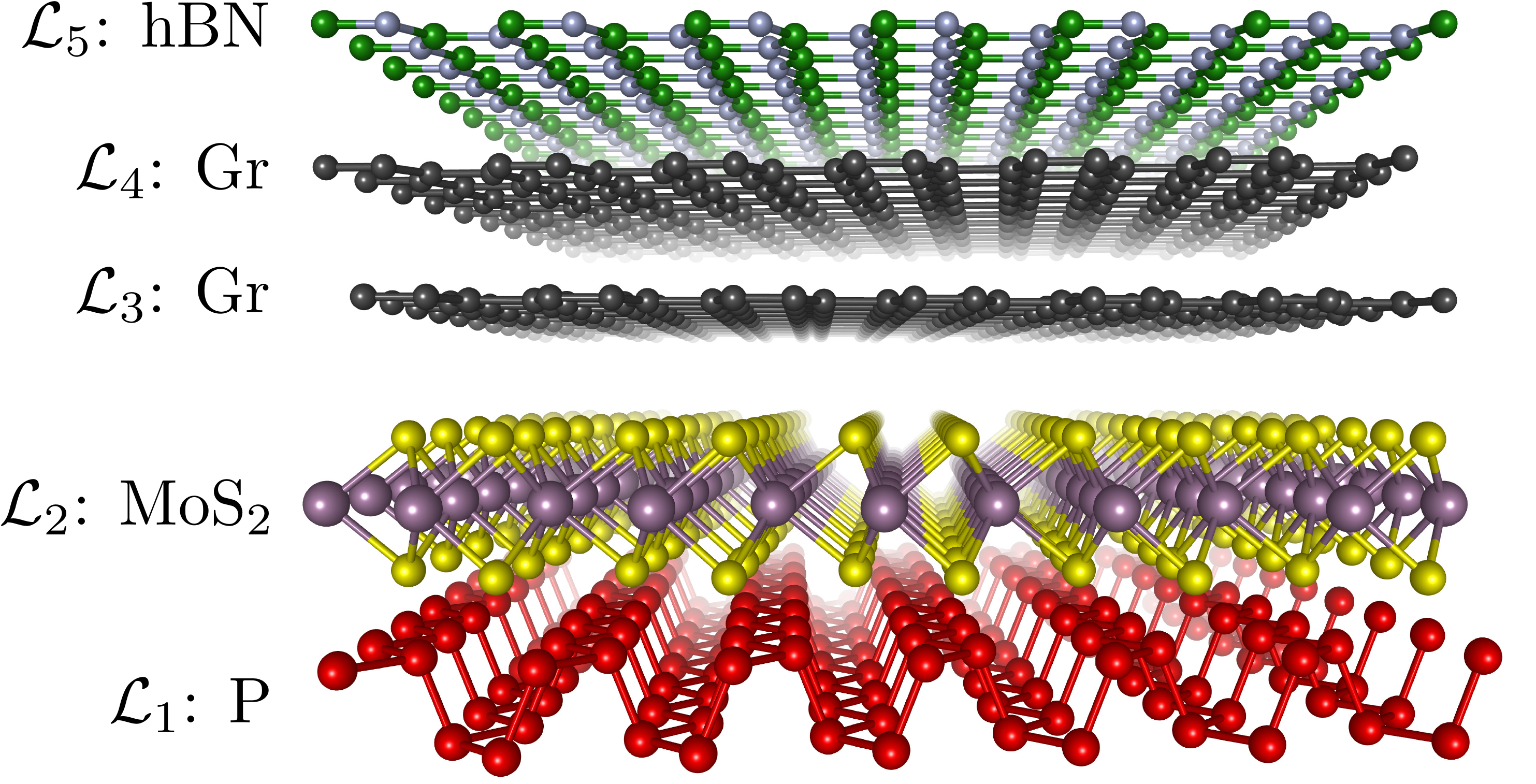}
	\caption{Example of a heterostructure: stack of five monolayers}\label{fig:figure1}
\end{figure}
Heterostructures are vertical stacks of a few two-dimensional crystalline monolayers.
An example is depicted in Figure~\ref{fig:figure1}. Five distinct layers with different atomic components and structures are positioned as a vertical stack.
Due to the weak van der Waals nature of the interactions between these layers, they do not relax into a common periodic structure, but rather each layer essentially keeps the structure it possesses as an isolated monolayer.
The resulting assembly is thus in general not periodic: it is an aperiodic structure with a long-range order.
A systematic model of heterostructures starts naturally by a rigorous depiction of this particular geometry, which can be idealized in the following way.

For the sake of generality, we consider $d$-dimensional systems embedded in $(d+1)$-dimensional space. Note that we choose $d$ as the natural dimension of the structure, since the $(d+1)$\textsuperscript{st} dimension plays a very particular and limited role, especially in tight-binding models. We are in particular interested in $d = 2$ for the layered heterostructures which motivate our study, while the choice $d = 1$ enables us to present simple numerical examples in Section~\ref{sec:numerics}.

Let us consider such a $(d+1)$-dimensional system of $p$ parallel $d$-dimensional periodic atomic layers denoted $\cL_j \subset \R^{d+1}$, $j = 1 \dots p$.
We denote by
\begin{itemize}
\item $(\be_1,\cdots,\be_{d+1})$ an orthonormal basis of the physical space such that each layer is perpendicular to $\be_{d+1}$; from now on we identify the physical space with $\R^{d+1}$ using the cartesian coordinates $\bx=(x_1,\cdots,x_{d+1})^T$ associated with this basis;
\item $h_j$ the $(d+1)$\textsuperscript{st} coordinate of the center of layer $j$. Without loss of generality, we can assume that $0=h_1< h_2 < \cdots < h_p$;
\item $\cR_j$ the $d$-dimensional periodic lattice of layer $j$;
\item $\bE_j$ the matrix in $\R^{d \times d}$ whose columns form a basis generating the lattice $\cR_j$:
\begin{equation}
\cR_j := \bE_j \Z^d \subset \R^d;
\end{equation}
\item $\Gamma_j := \R^d / \cR_j$ the quotient of $\R^d$ by the discrete lattice $\cR_j$, which has the topology of a $d$-dimensional torus and can be canonically identified with the periodic unit cell $\widehat\Gamma_j:=\bE_j  [-1/2,1/2)^d$ of layer $j$;
\item $m_j$ the motif of layer $j$, i.e., the measure on $\R^{d+1}$ supported in $\widehat\Gamma_j \times \R$ and representing the nuclear distribution $\rho^{\rm nuc}_j$ of layer $j$. More precisely, $m_j$ is a finite sum of positively weighted Dirac measures of the form
\[
	m_j = \sum_{k=1}^{M_j} z_k^{(j)} \delta_{\bx_k^{(j)}},
\]
where $M_j$ is the number of nuclei per unit cell in layer $j$, $z_1^{(j)}, \cdots, z_{M_j}^{(j)}$ the atomic charges of these nuclei, and $\bx_k^{(j)} \in \widehat\Gamma_j\times\R$ their positions in a reference configuration in which the center of layer $j$ belongs to the plane $x_{d+1}=0$. The nuclear distribution $\rho^{\rm nuc}_j$ of layer $j$ is then given by
\[
\rho^{\rm nuc}_j= \sum_{\bn_j \in \cR_j} m_j \left (\cdot - (\bgamma_j+\bn_j + h_j \be_{d+1})\right )= \sum_{\bn_j \in \cR_j} \sum_{k=1}^{M_j} z_k^{(j)} \delta_{\bgamma_j+\bn_j+\bx_k^{(j)}+h_j\be_{d+1}},
\]
where $\bgamma_j \in \widehat\Gamma_j$ depends on the horizontal position of the lattice sites of layer $j$ relatively to the origin of the coordinates. Here and in the sequel, we use the same notation to denote a vector $\by$ of $\R^d$ and its canonical embedding $(\by^T,0)^T$ in $\R^{d+1}$.
\end{itemize}
\begin{remark}
	Note that, by virtue of the identification $\Gamma_j \equiv \widehat\Gamma_j$, an equivalence class $\bgamma_j \in \Gamma_j$ can be seen either as a discrete set of points in $\R^d$, or as one point of the periodic unit cell $\widehat\Gamma_j$. Since both viewpoints are employed here, we denote these two usages differently to avoid confusion:
	\begin{itemize}
		\item $\bgamma_j$ will be used to denote a single point of the periodic unit cell $\widehat\Gamma_j$,
		\item $\bgamma_j + \cR_j$ will be used to denote the corresponding set of points in $\R^d$. In particular, the sum of the values of a function $f$ over the lattice sites of layer $j$ will be denoted as $\sum_{\bp_j \in \bgamma_j + \cR_j} f(\bp_j)$.
	\end{itemize}
\end{remark}

\subsection{Translation group and the hull}\label{sec:hullcontinuous}

To understand the geometry of our heterostructures, it is a useful exercise to picture the origin of coordinates as our viewpoint (in the sense of ''position of observation"). Neighboring atoms then constitute a local environment, e.g., A-A stacking (aligned layers) vs. A-B stacking (staggered layers) in a graphene bilayer~\cite{Fang2016}. This environment, i.e., the positions of all atoms relative to the origin, constitutes a choice of {\it configuration} for the structure.

In the case of a periodic material (a perfect crystal), the set of such possible configurations, or the {\it hull}~\cite{bellissard2003coherent}, is simply the periodic unit cell. Indeed, choosing the origin at points which differ only by a lattice vector results in identical configurations.
This invariance by lattice translations is what enables the classical use of the Bloch theorem to reduce the Hamiltonian operator set on the whole space to a family of easily analyzed operators on the periodic unit cell, indexed by the quasimomenta $\bk$ belonging to the Brillouin zone of the crystal.

This is not the case for the generically incommensurate layered structures presented in Section~\ref{subsec:geometry}. The change of coordinate origin (our viewpoint) is naturally associated with the action of the group $\R^d$ on $\R^{d+1}$ by translations $\mathtt{T}$ that are parallel to the layers:
\begin{equation}\label{def:Rdp1translations}
	\text{For } \mathbf{a} \in \R^d, \qquad
		{\mathtt T}_\mathbf{a}:  \left \{\begin{aligned}
			\R^{d+1} &\to \R^{d+1}, \\
			\mathbf{x} &\mapsto \left(x_1 + a_1, \dots, x_d + a_d, x_{d+1} \right).
		\end{aligned}\right.
\end{equation}
\begin{remark}\label{rem:continuousvsdiscrete}
	Note that there is no translational symmetry in the perpendicular $(d+1)$\textsuperscript{st}-dimension, and therefore we need not include this direction in the continuous translation group. The situation is different in discrete (tight-binding) models, where we model hopping from layer to layer as will be seen later on.
\end{remark}

Now we proceed with the formal definition of the hull. The positions of all the atoms are encoded in the nuclear charge distribution, a Radon measure in $\mathfrak{M}(\R^{d+1})$~\cite{bellissard2003coherent}, on which the translation group $\R^d$ acts naturally\footnote{Given $\mathbf{a} \in \R^d$, the translation $\mathtt{T}_\mathbf{a}$ acts on the space of continuous functions with compact support $\mathcal{C}_c(\R^{d+1})$ through $\mathtt{T}_\mathbf{a} f(\mathbf{x}) = f(\mathtt{T}_{-\mathbf{a}} \mathbf{x})$. Therefore it acts on the space $\mathfrak{M}(\R^{d+1})$ of Radon measures through $\mathtt{T}_\mathbf{a} \mu(f) = \mu \left (\mathtt{T}_{-\mathbf{a}} f \right )$ whenever $f \in \mathcal{C}_c(\R^{d+1})$ and $\mu \in \mathfrak{M}(\R^{d+1})$.}. Then the hull is the dynamical system $(\Omega, \R^d, \mathtt{T})$, where $\Omega$ is the closure of the orbit of the nuclear charge distribution measure on $\R^{d+1}$ generated by the atoms of all $p$ layers under the action of $\R^d$ through $\mathtt{T}$.
\begin{figure}[t]
\centering
	\begin{subfigure}[b]{.49\linewidth}
		\centering
		\includegraphics[width=.85\textwidth]{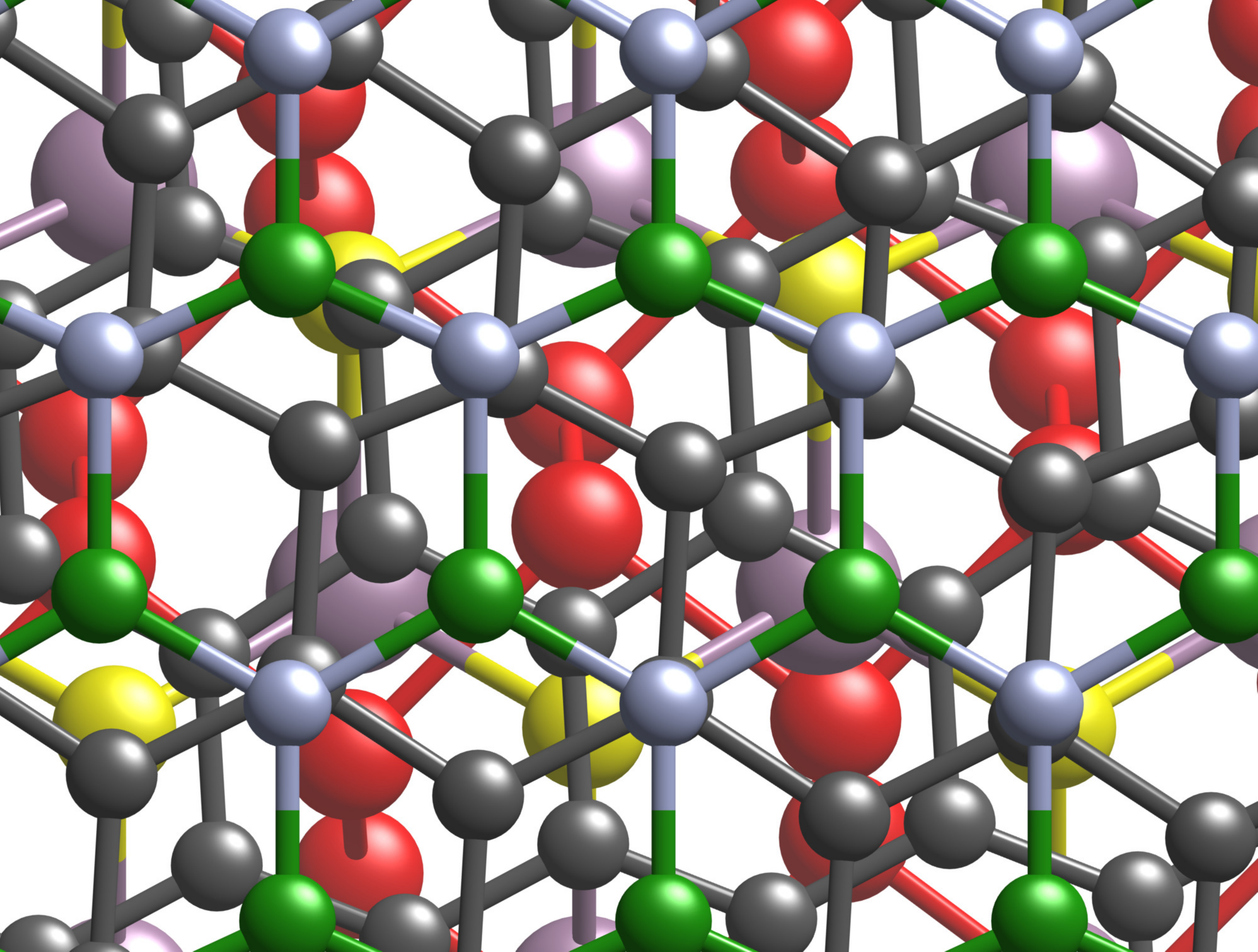}\\
		\vspace{.5cm}
		\caption{Top view of the crystalline structure}\label{fig:figure2a}
	\end{subfigure}
	\begin{subfigure}[b]{.49\linewidth}
		\centering
		\includegraphics[width=.8\textwidth]{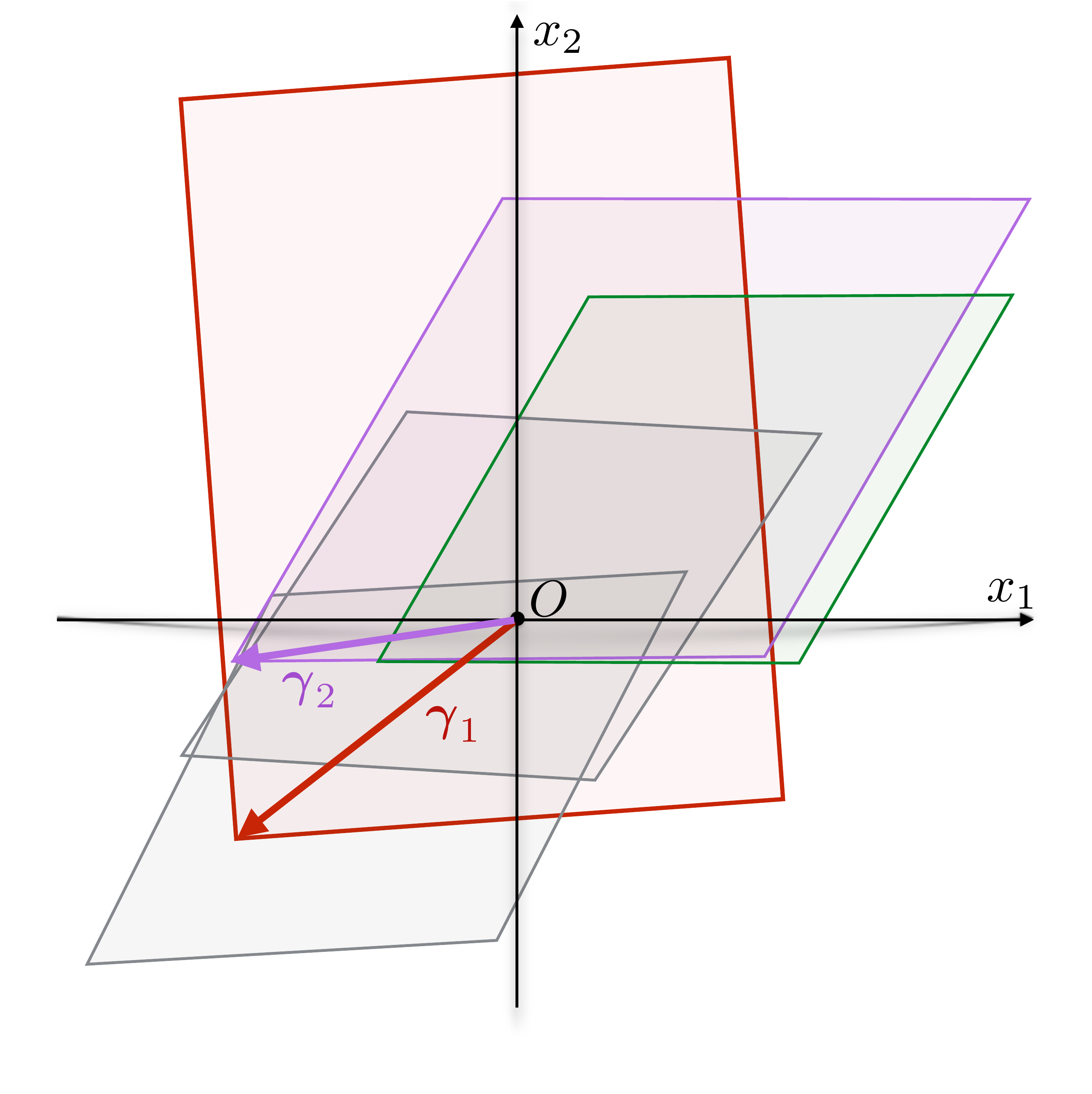}
		\caption{Projection of the origin into the unit cells}\label{fig:figure2b}
	\end{subfigure}
	\caption{Illustration of the process for labeling the local configurations of heterostructures. The vertices of the parallelograms are sites of the shifted lattices $\gamma_j + \cR_j$.}\label{fig:figure2}
\end{figure}

Parameterizing this orbit is similar to describing the position of all atoms relative to the origin, given an arbitrary translation of the system. An example of this process is presented by Figure~\ref{fig:figure2} (see also Figure~\ref{fig:figure3} for a simple one-dimensional picture). While the initial view of all atomic positions in Figure~\ref{fig:figure2a} might appear quite chaotic, it can actually be efficiently encoded. Since each individual layer $\mathcal{L}_j$, $j = 1, \dots, p$ is periodic, the set of all possible configurations for layer $j$ is in one-to-one correspondence with $\Gamma_j$. Indeed, an element $\bgamma_j \in \Gamma_j$ is the projection of the $j$-th layer periodic lattice on the horizontal plane, i.e., an equivalence class modulo $\cR_j$, as seen in Figure~\ref{fig:figure2b}. The overall configuration can thus be parameterized as an element of
\begin{equation}
	\Omega = \Gamma_1 \times \cdots \times \Gamma_p.
\end{equation}

\begin{figure}[t]
\centering
\includegraphics[width=.7\textwidth]{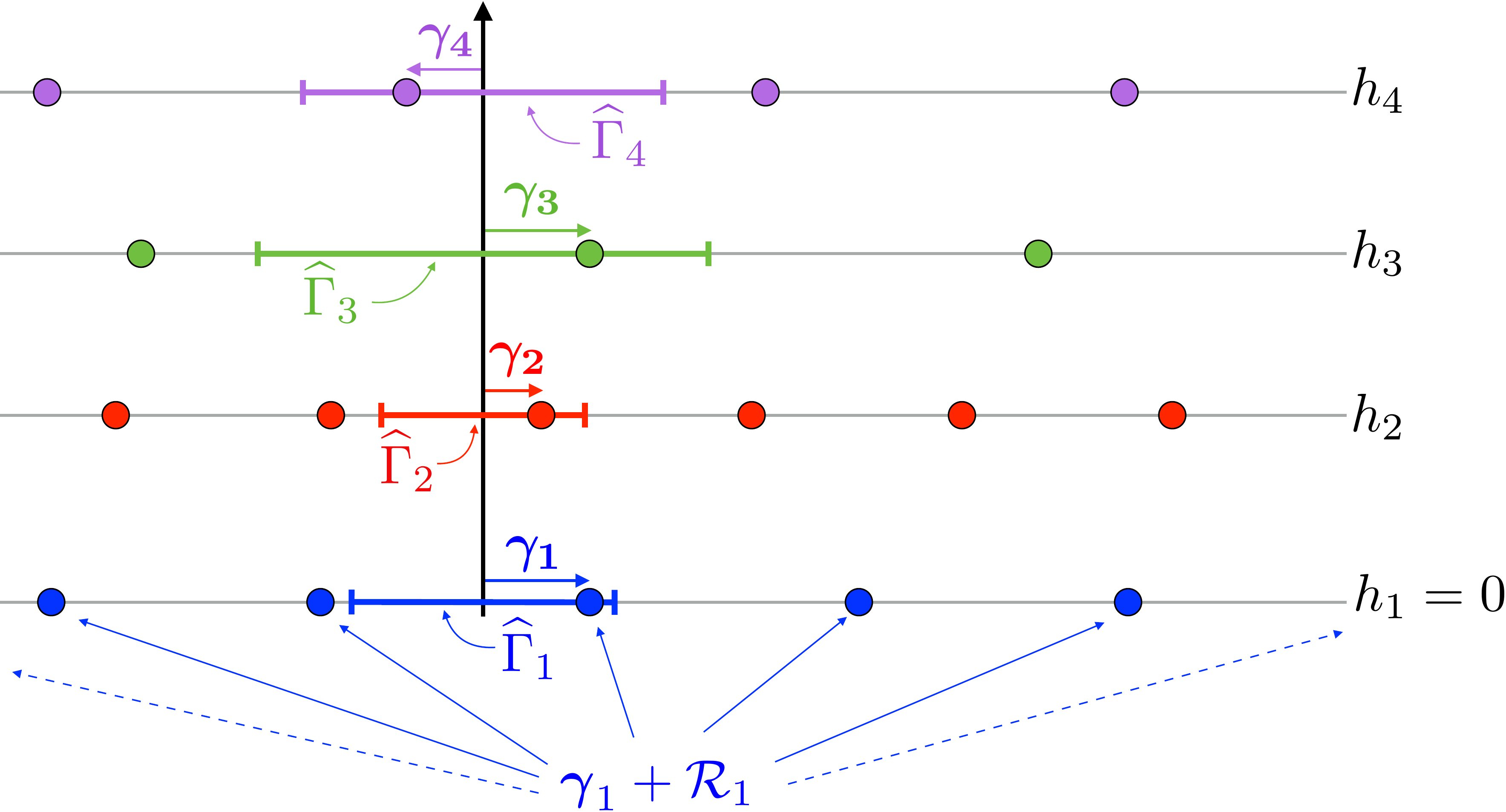}
\caption{Schematic representation of a configuration $\omega=(\bgamma_1,\bgamma_2,\bgamma_3,\bgamma_4) \in \Omega$ for $d=1$ and $p=4$. The vector $\bgamma_j$ on the figure corresponds to the unique representative of the equivalence class $\bgamma_j$ in the periodic unit cell $\widehat\Gamma_j$ and $\bgamma_1+\cR_1$ is the set of the lattice points of the first layer in the configuration $\omega$.}
\label{fig:figure3}
\end{figure}
The hull $\Omega$ therefore has the topology of a $(d p)$-dimensional torus.
For any given configuration $\omega = (\bgamma_1, \dots, \bgamma_p)$, the nuclear charge distribution is
\[
	\rho^{\rm nuc}_\omega = \sum_{j=1}^p \left ( \sum_{\bp_j \in \bgamma_j + \cR_j} \mathtt{T}_{\bp_j} m_j( \cdot - h_j \be_{d+1}) \right ).
\]
The support of the measure $\rho^{\rm nuc}_\omega$ is the discrete set:
\begin{equation}\label{def:parameterizationatomicset}
	\left \{ \begin{aligned}
	&\cL^{\omega} := \bigcup_{j = 1}^p \cL_j^{\omega} \subset \R^{d+1} &\text{ (set of all atomic positions)},\\
	&\text{ with }\cL_j^{\omega} := \textrm{Supp}(m_j) + \bgamma_j + \cR_j +  h_j \be_{d+1}  &\text{ (atomic positions in layer $j$)}.
	\end{aligned} \right.
\end{equation}
There exists a natural action of the additive $\R^d$ group on $\Omega$, corresponding to an horizontal translation of the whole system in the previous parameterization, defined by
\begin{equation}\label{def:omegatranslations}
	\mathtt{T}_\ba (\bgamma_1, \dots, \bgamma_p) = \left( \bgamma_1 + \ba , \dots, \bgamma_p + \ba \right ).
\end{equation}
For convenience, we use the same notation ${\mathtt T}$ to denote the action of the translation group $\R^d$ on $\R^{d+1}$, formula \eqref{def:Rdp1translations}, and on $\Omega$, formula~\eqref{def:omegatranslations}.

Finally, the dynamical system $(\Omega, \R^d, \mathtt{T})$ is equipped with a probability measure $\Prob$ which encodes the relative occurrences of the various configurations. This measure should be invariant by the translation action to reflect the spatial homogeneity of the system. There is here a unique such measure under the incommensurability condition given in the following definition. Let us denote the dual (or reciprocal) lattice of any cocompact\footnote{Recall that a cocompact lattice $\cR$ of $\R^d$ is a discrete periodic lattice of $\R^d$ with $d$ free vectors (in such a way that $\R^d/\cR$ is compact), and that the dual lattice of a cocompact lattice $\cR$ is defined by
\[
	\cR^* = \left \{ \bk \in \R^d \ \vert \ \bk \cdot \bn \in \Z, \quad \forall \bn \in \cR \right \}.
\]
When $\cR$ is generated by the columns of a matrix $\mathbf{E}$, its dual $\cR^*$ is generated by the columns of $(\mathbf{E}^T)^{-1}$, the inverse transpose of $\mathbf{E}$.} 
lattice $\cR$ by $\cR^*$.
\begin{definition}\label{definition:incommensurability}
	The collection of cocompact lattices $\cR_1, \cdots , \cR_p$ of $\R^d$ is called incommensurate if we have for any $p$-tuple $(\bk_1, \dots, \bk_p) \in \cR_1^* \times \cdots \times \cR_p^*$,
	\begin{equation}\label{def:incommensurability}
		\sum_{j = 1}^p \bk_j = \mathbf{0} \quad \Leftrightarrow \quad \bk_j = \mathbf{0} \ \ \forall  j = 1, \dots, p.
	\end{equation}
\end{definition}
This definition, surprising at first, defines incommensurability as the absence of constructive interferences, or Bragg reflections, between the lattices.
The following result proves that it is also the right condition for the layered system to have a {\em homogeneous} character, i.e., all possible configurations will be visited uniformly as we translate our viewpoint along the horizontal place.
\begin{proposition} \label{prop:continuousHullErgodicity}
Let $\cR_1, \cdots , \cR_p$ of $\R^d$ be cocompact lattices of $\R^d$, $\Gamma_j = \R^d/\cR_j$, and $\Omega=\Gamma_1 \times \cdots \times \Gamma_p$, endowed with the uniform probability measure $\Prob$. Then,
	\begin{enumerate}
	\item $\Prob$ is invariant by the translation group $\R^d$;
	\item the dynamical system $(\Omega,\R^d,\mathtt{T},\Prob)$ is uniquely ergodic if and only if the lattices $\cR_1, \cdots , \cR_p$ are incommensurate. In this case, we have the Birkhoff property: for any $f \in C(\Omega)$ and $\omega \in \Omega$,
	\begin{equation}\label{def:continuousBirkhoff}
		\lim_{r \to \infty} \frac{1}{\vert B_r \vert } \int_{B_r} f(\mathtt{T}_{-\ba} \omega) \mathrm{d}\ba = \int_\Omega f \mathrm{d}\Prob,
	\end{equation}
	where $B_r$ is the ball of radius $r$ centered at zero.
\end{enumerate}
\end{proposition}
\begin{proof}
	The first point follows from the definition~\eqref{def:omegatranslations} and the translation invariance of the Lebesgue measure. To show that incommensurability implies ergodicity, we will make use of the following lemma:
	\begin{lemma}\label{lem:ergodicity1}
		Let $\cR$ be a lattice in $\R^d$ and $\mathbf{k} \in \R^d$. Then:
		\begin{equation}\label{eq:ergodicitylemma}
			 \lim_{r \to \infty} \frac{1}{\# \left ( \cR \cap B_r \right )} \sum_{\bn \in \cR \cap B_r} e^{2i\pi \bk \cdot \bn} = 0 \quad \Leftrightarrow \quad  \bk \notin \cR^*.
		\end{equation}
	\end{lemma}
	It is clear that if $ \bk \in \cR^*$, then the weighted sum in the left-hand side in~\eqref{eq:ergodicitylemma} is always equal to $1$ for all $r>0$, and thus the limit is not zero.
	Let $(\bc_1, \dots, \bc_d)$ be a basis of $\cR$ and $(\bc^*_1, \dots, \bc^*_d)$ the associated dual basis of $\cR^*$ ($\bc_j \cdot \bc_k^* = \delta_{jk}$).
	We assume now that $\bk \notin \cR^*$, and expanding $\bk = k_1 \bc_1^* + \dots + k_d \bc_d^*$, we suppose that $k_1 \notin \Z$ without loss of generality.

	Let $\mathcal{E}_r^\perp$ be the projection of $\cR \cap B_r$ onto the last $(d-1)$ lattice coordinates:
	\[
		\mathcal{E}_r^\perp = \left \{(n_2, \dots, n_d) \in \Z^{d-1} \quad \biggr\vert \quad \exists n_1 \in \Z,\ \sum_{j=1}^d n_j \bc_j \in \cR \cap B_r \right \}.
	\]
	There exists for any $(d-1)$-tuple $\widetilde{\bn} = (n_2, \dots, n_d)$ in $\mathcal{E}_r^\perp$ two integers $N_r^\pm(\widetilde{\bn})$ such that for any $n_1 \in \Z$, $\sum_{j = 1}^d n_j \bc_j$ belongs to $\cR \cap B_r$ if and only if $N_r^-(\widetilde{\bn}) \leq n_1 < N_r^+(\widetilde{\bn})$.
	We decompose accordingly the sum in~\eqref{eq:ergodicitylemma}:
	\begin{align*}
		\sum_{\bn \in \cR \cap B_r} e^{2i\pi \bk \cdot \bn} &= \sum_{\widetilde{\bn} \in \mathcal{E}_r^\perp} \sum_{n_1 = N_r^-(\widetilde{\bn})}^{N_r^+(\widetilde{\bn})-1}  e^{2 i \pi \left (k_1 n_1  + \widetilde{\bk} \cdot \widetilde{\bn} \right )}\\
		&= \sum_{\widetilde{\bn} \in \mathcal{E}_r^\perp} e^{2 i \pi \widetilde{\bk} \cdot \widetilde{\bn} } \left ( \frac{e^{2i\pi k_1 N_r^-(\widetilde{\bn})} - e^{2i\pi k_1 N_r^+(\widetilde{\bn})}}{1 - e^{2i\pi k_1}} \right ).
	\end{align*}
	We can bound the number of elements in $\left ( \cR \cap B_r \right )$ from below by $c_1 r^d$ and of $\mathcal{E}_r^\perp$ by $c_2 r^{d-1}$ where $c_1, c_2 > 0$ are two geometrical constants depending only on $\cR$. Thus
	\[
		\left \vert \frac{1}{\# \left ( \cR \cap B_r \right )} \sum_{\bn \in \cR \cap B_r} e^{2i\pi \bk \cdot \bn} \right \vert \leq \left ( \frac{2}{\vert 1 - e^{2i\pi k_1}\vert}\right ) \frac{c_2}{c_1 r}.
	\]
	Since $k_1 \notin \Z$, Lemma~\ref{lem:ergodicity1} is proved.

	{\bf Birkhoff property.} Let us now suppose that the collection of lattices $\cR_1, \cdots , \cR_p$ is incommensurate. Let $f \in C(\Omega)$ and $\omega = (\bgamma_1, \dots, \bgamma_p) \in \Omega$.  Let $\varepsilon > 0$ and $T_\varepsilon$ be a trigonometric polynomial such that $\Vert f - T_\varepsilon \Vert_\infty \leq \varepsilon$. $T_\varepsilon$ is a finite linear combination of Fourier factors of the form
	\[
		G_{\bk_1, \dots, \bk_p}: (\bgamma_1, \dots, \bgamma_p) \mapsto e^{2i\pi (\bk_1 \cdot \bgamma_1 + \cdots + \bk_p \cdot \bgamma_p)} \quad \text{ with } \quad (\bk_1, \dots, \bk_p) \in \cR_1^* \times \cdots \times \cR_p^*.
	\]
	If $\bk_j = \bzero$ for all $j = 1, \dots, p$, then clearly
	\[
		\lim_{r \to \infty} \frac{1}{\vert B_r \vert } \int_{B_r} G_{\bzero, \dots, \bzero}(\mathtt{T}_{-\ba} \omega) \mathrm{d}\ba = 1 = \int_\Omega G_{\bzero, \dots, \bzero} \mathrm{d}\Prob.
	\]
	Otherwise, we assume without loss of generality that $\bk_2 \neq \bzero$, and we approximate the ball $B_r$ for $r>0$ by a union of translated unit cells,
	\[
		B^1_r = \bigcup_{\bn_1 \in \cR_1 \cap B_r} \left ( \widehat{\Gamma}_1  + \bn_1 \right ) \qquad \text{where }  \widehat{\Gamma}_1 = \mathbf{E}_1[-1/2, 1/2)^d.
	\]
	Averaging over this approximate ball, we find:
	\begin{align*}
		\frac{1}{ \vert B_r^1 \vert } & \int_{B^1_r} G_{\bk_1, \dots, \bk_p}(\mathtt{T}_{-\ba} \omega) \mathrm{d}\ba = \frac{\vert \widehat{\Gamma}_1 \vert^{-1}}{\# \left ( \cR_1 \cap B_r \right )}\sum_{\bn_1 \in \cR_1 \cap B_r} \int_{\widehat{\Gamma}_1} e^{2 i \pi \sum_{j = 1}^p \bk_j \cdot (\bgamma_j - \ba - \bn_1)}\mathrm{d}\ba\\
		&= \left ( \vert \widehat{\Gamma}_1 \vert^{-1} \int_{\widehat{\Gamma}_1} e^{2 i \pi \sum_{j = 1}^p \bk_j \cdot (\bgamma_j - \ba)}\mathrm{d}\ba \right ) \frac{1}{\# \left ( \cR_1 \cap B_r \right )} \sum_{\bn_1 \in \cR_1 \cap B_r} e^{-2 i \pi \sum_{j = 2}^p \bk_j \cdot \bn_1}.
	\end{align*}
	We assume that the lattices $\cR_1, \cdots , \cR_p$ are incommensurate, and thus $\sum_{j = 2}^p \bk_j \in \cR_1^*$ if and only if $\bk_j = \bzero$ for all $j = 2, \dots, p$.
	This is not possible since $\bk_2 \neq \bzero$, and we therefore deduce from Lemma~\ref{lem:ergodicity1} that, uniformly in $\omega$,
	\[
		\lim_{r \to \infty} \frac{1}{ \vert B_r^1 \vert } \int_{B^1_r} G_{\bk_1, \dots, \bk_p}(\mathtt{T}_{-\ba} \omega) \mathrm{d}\ba = 0.
	\]
	Note that for some $C>0$ independent of $r$, we have $\vert B_r \Delta B^1_r \vert \leq C r^{d-1}$ (where $A \Delta B$ denotes the symmetric difference of the sets $A$ and $B$). Thus we conclude
	\[
		\forall (\bk_1, \dots, \bk_p) \in (\cR_1^* \times \cdots \times \cR_p^*) \setminus \left\{{\mathbf 0}\right\}, \quad \lim_{r \to \infty} \frac{1}{\vert B_r \vert } \int_{B_r} G_{\bk_1, \dots, \bk_p}(\mathtt{T}_{-\ba} \omega) \mathrm{d}\ba = 0, \text{ uniformly in } \omega.
	\]
	As a result, the trigonometric polynomial $T_\varepsilon$ satisfies for $r>0$ large enough, uniformly in $\omega$,
	\[
		\left \vert \frac{1}{\vert B_r \vert } \int_{B_r} T_\varepsilon(\mathtt{T}_{-\ba} \omega) \mathrm{d}\ba - \int_\Omega T_\varepsilon \mathrm{d} \Prob\right \vert \leq \varepsilon.
	\]
	Finally, since $\Vert f - T_\varepsilon \Vert_\infty \leq \varepsilon$ we obtain the Birkhoff ergodic formula: for $r$ large enough, uniformly in $\omega$,
	\begin{equation}\label{eq:proofBirkhoffcontinuous}
		\left \vert \frac{1}{\vert B_r \vert } \int_{B_r} f(\mathtt{T}_{-\ba} \omega) \mathrm{d}\ba - \int_\Omega f \mathrm{d} \Prob\right \vert \leq 3 \varepsilon.
	\end{equation}

	\noindent {\bf Ergodicity.}
	Let $B$ be a measurable subset of $\Omega$ invariant under the action of $\R^d$. Let $\varepsilon > 0$ and $f_\varepsilon$ be a continuous function such that $\Vert \chi_B - f_\varepsilon\Vert_{L^1(\Omega)} \leq \varepsilon$. By invariance of $B$ under translations, we have for all $\ba \in \R^d$,
	\[
		\Vert f_\varepsilon \circ \mathtt{T}_\ba - f_\varepsilon \Vert_{L^1(\Omega)} \leq 2 \varepsilon.
	\]
	We can bound the difference between $f_\varepsilon$ and its Birkhoff means:
	\begin{align*}
		\left \Vert \frac{1}{\vert B_r \vert } \int_{B_r} f_\varepsilon \circ \mathtt{T}_{-\ba}  \mathrm{d}\ba - f_\varepsilon\right \Vert_{L^1(\Omega)} &= \int_\Omega \left \vert \frac{1}{\vert B_r \vert }  \int_{B_r} f_\varepsilon( \mathtt{T}_{-\ba} \omega ) \mathrm{d}\ba - f_\varepsilon(\omega) \right \vert \mathrm{d}\Prob(\omega)\\
		& \leq \frac{1}{\vert B_r \vert }  \int_{B_r} \int_\Omega \left \vert f_\varepsilon ( \mathtt{T}_{-\ba} \omega ) - f_\varepsilon(\omega) \right  \vert \mathrm{d}\Prob(\omega)  \mathrm{d}\ba \leq 2 \varepsilon,
	\end{align*}
	where we have used the triangle inequality for integrals and Fubini's theorem for non-negative functions. Since the Birkhoff means of $f_\varepsilon$ converge uniformly in $\omega$ for $r \to \infty$ to $\int_\Omega f_\varepsilon \mathrm{d} \Prob$, we deduce that
	\[
		\left \Vert \int_\Omega f_\varepsilon \mathrm{d} \Prob - f_\varepsilon\right \Vert_{L^1(\Omega)} \leq 2 \varepsilon.
	\]
	As a consequence, we have by the triangle inequality,
	\begin{align*}
		\Vert \chi_B - \Prob(B) \Vert_{L^1(\Omega)} \leq \Vert \chi_B - f_\varepsilon \Vert_{L^1(\Omega)} + \left \Vert f_\varepsilon -  \int_\Omega f_\varepsilon \mathrm{d} \Prob  \right \Vert_{L^1(\Omega)} + \left \Vert \int_\Omega f_\varepsilon \mathrm{d} \Prob  - \Prob(B) \right \Vert_{L^1(\Omega)} \leq 4 \varepsilon.
	\end{align*}
	Since this holds for any $\varepsilon > 0$, we conclude that $\chi_B$ is constant a.s., and therefore $\Prob(B) \in \{0,1\}$. Thus $\Prob$ is ergodic.

	\noindent {\bf Non-averaging case.} If the lattices $\cR_1, \dots , \cR_p$ are not incommensurate, there exists a particular nonzero combination $(\bk_1, \dots, \bk_p) \in \cR_1^* \times \cdots \times \cR_p^*$ such that $\sum_{j = 1}^p \bk_j = \bzero$. Then, we have for the corresponding Fourier factor,
	\begin{align*}
		G_{\bk_1, \dots, \bk_p}(\mathtt{T}_{-\ba} \omega) = G_{\bk_1, \dots, \bk_p}(\omega), \qquad \forall (\ba, \omega) \in \R^d \times \Omega .
	\end{align*}
	The function $G_{\bk_1, \dots, \bk_p}$ is then invariant by the $\R^d$-action and not constant. Therefore the dynamical system $(\Omega,\R^d,\mathtt{T},\Prob)$ is not ergodic in this case.
\end{proof}
\begin{remark}
	When the lattices $\cR_1, \cdots , \cR_p$ are {\it commensurate}, i.e., $\bigcap_{j = 1}^p \cR_j$ forms a cocompact superlattice of $\R^d$, then the orbits through the action $\mathtt{T}$ of $\R^d$ on $\Omega$ are no longer dense in $\Omega$, rather they form lower dimensional submanifolds of $\Omega$ which are then distinct hulls for the possible configurations. In this case, Bloch theory allows one to study each of these nonequivalent configuration sets individually.
\end{remark}
\begin{remark}
	Note that the lattices $\cR_1, \cdots , \cR_p$ can also be neither commensurate nor incommensurate. We do not know if one can always find a sensible parameterization of the hull of a particular configuration in this case, since its orbit under $\R^d$ may not be dense in $\Omega$, nor has a simple geometry.
\end{remark}

\subsection{Disordered multilayer systems}\label{subsec:disorder}
The case of disordered multilayer systems is more involved. For the sake of both brevity and clarity, we consider the specific, but representative, example of a bilayer system whose bottom layer lays on a periodic substrate modeled by an external $\cR_0$-periodic potential, and whose top layer may have defects. We denote by $\cR_0$ the periodic lattice of the substrate, and by $\cR_1$ and $\cR_2$ the periodic lattices of the bottom and top layers respectively, and we assume that $\cR_0$, $\cR_1$ and $\cR_2$ are incommensurate in the sense of Definition~\ref{definition:incommensurability}. We then assume that the defects are such that
\begin{itemize}
\item the overall geometry of the system is not modified: the periodic lattice $\cR_2$ is still appropriate to describe the configurations of the system; the difference with the
case of "perfect" homogeneous systems dealt with in the previous section, is that the nuclear distribution inside the unit cells of the top layer is not the same in each cell;
\item in each cell in the top layer, the motif can be one of the following: $(m_2^{(0)},m_2^{(1)},\cdots,m_2^{(D)})$,
where $m_2^{(0)}$ is the periodic motif of the top layer in the absence of defects, and $m_2^{(1)},\cdots,m_2^{(D)}$ correspond to the different kinds of defects that can be observed.
For instance, in the case of a graphene layer for which each carbon atom can adsorb a hydrogen atom, there are two carbon atoms in each cell, $m_2^{(0)}$ corresponds to the case when no hydrogen atom is adsorbed, $m_2^{(1)}$ (respectively, $m_2^{(2)}$) to the case when only the first (respectively, second) carbon atom has adsorbed a hydrogen atom, and $m_2^{(3)}$ to the case when the two carbon atoms have adsorbed hydrogen atoms;
\item the defects are independently and identically distributed in the cells of the top layer. We denote by $(p_0,p_1,\cdots,p_D)$ (with $p_j > 0$ and $\sum_{j=0}^D p_j=1$) the probability distribution of the motifs $(m_2^{(0)},m_2^{(1)},\cdots,m_2^{(D)})$.
\end{itemize}
Without loss of generality, we can assume that the center of the bottom layer is contained in the plane $x_{d+1}=0$, and that the center of the top layer is contained in the plane $x_{d+1}=h >0$. For each point in $\br \in \R^d$, we introduce the decomposition of $\br$ associated with the lattice $\cR_2$ defined as
\begin{equation}\label{def:decomposition}
\br = \left [ \br \right ]_2 + \left \{ \br \right \}_2 \qquad \text{ with } \left [ \br \right ]_2 \in \cR_2 \text{ and } \left \{ \br \right \}_2 \in \widehat\Gamma_2=\bE_2 [-1/2, 1/2)^d.
\end{equation}

For this example, the hull is the dynamical system $(\Omega,\R^d,\tau,\Prob)$ where
\begin{itemize}
\item $\Omega= \Gamma_0 \times \Gamma_1 \times \Gamma_2 \times \left\{ 0,1,\cdots,D \right\}^{\cR_2}$, the configuration $\omega =(\bgamma_0,\bgamma_1,\bgamma_2,\left\{t_{\bn_2}\right\}_{\bn_2 \in \cR_2}) \in \Omega$ corresponding to the nuclear distribution
$$
\rho^{\rm nuc}_\omega= \sum_{\bp_1 \in \bgamma_1+\cR_1} T_{\bp_1}m_1 + \sum_{\bp_2 \in \bgamma_2+\cR_2} T_{\bp_2}m_2^{(t_{[\bp_2]_2})}(\cdot - h \be_{d+1}).
$$
Here $t_{[\bp_2]_2} \in \left\{ 0,1,\cdots,D \right\}$ is the type of the motif carried by the cell of the top layer centered at $\bp_2+h \be_{d+1}$. Note that $\gamma_0$ does not appear in the definition of $\rho^{\rm nuc}_\omega$, since we do not take into account the geometry relaxation of the bilayer system due to the interaction with the substrate. On the other hand, the substrate generates a $\cR_0$-periodic potential which modifies the electronic structure of the bilayer system, and this potential depends on $\gamma_0$, the relative position of the substrate with respect to the bilayer system;
\item $\tau$ is the action of the group $\R^d$ on $\Omega$ defined by
\begin{equation*}
\begin{aligned}
\forall \ba \in \R^d, \quad \forall &\omega=(\bgamma_0, \bgamma_1, \bgamma_2, \left\{ t_{\bn_2} \right\}_{\bn_2 \in \cR_2}), \\&\qquad
\tau_\ba(\omega)=(\bgamma_0+\ba, \bgamma_1+\ba, \bgamma_2+\ba, \left\{ t_{\bn_2-[\{\bgamma_2\}_2+\ba]_2} \right\}_{\bn_2 \in \cR_2});
\end{aligned}
\end{equation*}
\item denoting by $\mu$ the probability on the set $\left\{ 0,1,\cdots,D \right\}$ with law $(p_0,p_1,\cdots,p_D)$, the ergodic probability $\Prob$ on $\Omega$ is defined as
$$
\Prob = \Prob_0 \otimes \Prob_1 \otimes \Prob_2 \otimes \Prob_{\rm dis},
$$
where $\Prob_j$ is the uniform probability distribution on $\Gamma_j$, and $\Prob_{\rm dis}=\mu^{\otimes \cR_2}$ is the probability distribution on the disorder.
\end{itemize}

\section{C$^*$-algebra formalism for tight-binding models}
\label{sec:CSA}

In this section, we adapt to the tight-binding modeling of perfect multilayer atomic heterostructures the $C^*$-algebra formalism which was extensively used by Belissard and collaborators~\cite{bellissard1994noncommutative,bellissard2003coherent} to model and analyze transport in aperiodic solids.

In the framework of tight-binding models, it is appropriate (see Remark~\ref{rem:continuousvsdiscrete}) to use the following alternative definition of the hull:
\begin{itemize}
\item $\dps \Omega_{\rm D} =  \left\{1, \cdots, p \right\} \times \Omega$;
\item $\mathfrak{t}$ is the action of the group
$$
\mathbb{G}_{\rm D} =\Z/p\Z \times \R^d
$$
defined by: for all $a=(\alpha,\ba) \in \G_{\rm D}$, and all $(j,\bgamma_1,\cdots,\bgamma_p) \in \Omega_{\rm D}$,
\[
	\mathfrak{t}_{a} (j,\bgamma_1,\cdots,\bgamma_p) = ( j - \alpha,\bgamma_1+\ba,\cdots,\bgamma_p+\ba),
\]
where the $-$ sign in $j-\alpha$ refers to the natural action of $\Z/p\Z$ on $\left\{1, \cdots, p \right\}$ (if $\alpha=k +p\Z$, $j-\alpha=((j-k-1) \mod p)+1$);
\item $\Prob_{\rm D}$ is as usual the uniform probability distribution on $\Omega_{\rm D}$, and it is ergodic if and only if the lattices $\cR_1, \dots , \cR_p$ are incommensurate.
\end{itemize}
More explicitly, for non-disordered multilayers, $\Omega_{\rm D}$ consists of $p$ copies of $\Gamma_1 \times \cdots \times \Gamma_p$ indexed by $j \in \left\{1, \cdots, p \right\}$ such that the reference lattice site of the configuration is in layer $j$.

In the tight-binding representation, the quantum states are expanded on a finite number of orbitals in each periodic unit cell, for example a set of maximally localized Wannier orbitals~\cite{wannier1937structure,kohn1959analytic,marzari1997maximally}.
We therefore introduce, for each layer $k$, a finite set $\Xi_k$ of tight-binding orbitals per unit cell. For a given configuration $\omega = (j,\bgamma_1, \cdots, \bgamma_p)$, the tight-binding orbital of type $n \in \Xi_k$ associated with the lattice site of layer $k$ located at point $\bx+h_k \be_{d+1} \in \bgamma_k + \cR_k + h_k \be_{d+1}$ is indexed by the triplet $((k-j) \mod p,\bx,n) \in \Z/p\Z \times (\gamma_k+\cR_k) \times \Xi_k$. The integer $((k-j) \mod p)$, considered as a element of $\Z/p\Z$, accounts for the vertical jump (in terms of number of layers and modulo $p$) to go from layer $j$, which contains the origin in the configuration $\omega$, to layer $k$ (see also the graphical explanation in the one-dimensional case in Figure~\ref{fig:figure4}). This way of labeling the orbitals turns out to be well suited to the $C^*$-algebra formalism introduced in the following two subsections.

The infinite set $\Xi^\omega$ of all tight-biding orbitals of the $p$ layers is then indexed by:
\begin{equation}\label{def:DiscreteStates}
	\Xi^\omega = \bigcup_{k=1}^p   \{(k-j)\mod p\} \times (\bgamma_k + \cR_k) \times\Xi_k.
\end{equation}
In the special case when there is only one orbital per unit cell, each $\Xi_k$ only contains one element, and can therefore be omitted in the above definition of $\Xi^\omega$.

In the configuration $\omega$, the space of quantum states for the tight-binding model is
\begin{equation}\label{def:QuantumStates}
	\mathfrak{H}_\omega = \ell^2(\Xi^\omega),
\end{equation}
and observables such as the Hamiltonian are described as linear operators on $\mathfrak{H}_\omega$.

\subsection{Abstract setting}\label{subsec:abstractsetting}
Let us first briefly recall in this section the general formalism of groupoid $C^*$-algebras. We refer e.g. to~\cite{bellissard1994noncommutative} for a more in-depth presentation of these mathematical objects.
We further simplify the presentation by assuming that there is only one tight-binding orbital per unit cell in each layer, so that
\begin{equation}\label{eq:Xiomega}
\Xi^\omega \equiv \bigcup_{k=1}^p   \{(k-j)\mod p\} \times (\bgamma_k + \cR_k) \subset \G_D.
\end{equation}
Note that this is not a restriction of the formalism (see Remark~\ref{rem:MultipleOrbitals} below).

The first step is to construct a groupoid based on the canonical transversal $X$ of the hull $\Omega_D$:
\begin{equation}\label{def:transversal}
X = \left \{ \omega \in \Omega_D \ \vert\ \omega = (j,\bgamma_1, \dots, \bgamma_p); \; \bgamma_j = \bzero \right \}.
\end{equation}
The idea is that while $\Omega_D$ indexes the possible viewpoints from any point in the layer planes, each element of $X$ represents a possible unique viewpoint from the position of a lattice site, which is then chosen as the origin.
This is a more appropriate approach in the case of tight-binding models.
Associated with the transversal $X$ is the groupoid $\Gamma(X)$ defined as follows:
\begin{equation}\label{def:groupoid}
\Gamma(X) = \left \{ (\omega, a) \in X \times \G_D \; \vert \; \mathfrak{t}_{-a} \omega \in X \right \}.
\end{equation}
For $\omega=(j,\bgamma_1,\cdots,\bgamma_p) \in \Omega_D$ and $a=(\alpha,\ba) \in G_D$, $(\omega,a) \in \Gamma(X)$ if and only if $\bgamma_{j}=0$ and $\bgamma_{j+\alpha}=\ba$. Geometrically, $\Gamma(X)$ indexes all the possible jumps between two lattice sites in all possible configurations of the multilayer system. The physical space vector of the jump associated with $(\omega,a)$ is $\bA=\ba + (h_{j+\alpha}-h_j)\be_{d+1}$. The groupoid is equipped with a set of three operations, respectively the range $r:\ \Gamma(X) \to X$, source $s: \Gamma(X) \to X,$ and composition $\circ: \Gamma(X) \times \Gamma(X) \to \Gamma(X)$, satisfying:
\begin{equation}
r(\omega, a) = \omega, \qquad s(\omega, a) = \mathfrak{t}_{-a} \omega, \qquad (\omega, a) \circ (\mathfrak{t}_{-a} \omega, b) = (\omega, a + b).
\end{equation}
The fiber $\Gamma^{(\omega)}$ is defined as $r^{-1}(\{ \omega \})$ for any $\omega \in X$.\\

In a second step, we define the $^*$-algebra $\mathcal{A}_0$ of continuous functions with compact support defined on $\Gamma(X)$ with values in $\C$, endowed with the following composition laws and $^*$-operator:
\begin{equation}\label{def:algebra}
\begin{aligned}
(\lambda f ) (\omega, a) &= \lambda f(\omega, a), \\
(f + g ) (\omega, a) &= f(\omega, a) + g(\omega, a), \\
(f * g)(\omega, a) &= \sum_{(\omega, b) \in \Gamma^{(\omega)}} f(\omega, b) g(\mathfrak{t}_{-b} \omega, a - b),\\
f^*(\omega, a) &= \overline{f\left (\mathfrak{t}_{-a} \omega, -a\right )}.
\end{aligned}
\end{equation}
The $^*$-algebra $\mathcal{A}_0$ has an identity, denoted by $\mathbf{1}$:
\[
	\mathbf{1}(\omega, a) = \delta_{a,0}.
\]
\begin{remark}\label{rem:MultipleOrbitals}
	In the case when there are several atoms per unit cell, and/or each atom carries more than one orbital, the functions $f$ in \eqref{def:algebra} do not take their values in $\C$ but $f(\omega, a)$ is in $\C^{N_j \times N_{j+\alpha}}$ when $\omega = (j,\bgamma_1, \dots, \bgamma_p)$ and $a = (\alpha,\ba)$, where $N_k = \#(\Xi_k)$ is the total number of tight-binding orbitals carried by the atoms in the unit cell of layer $k$. Products of the form $f(\omega, b) g(\mathfrak{t}_{-b} \omega, a - b)$ should then be understood as matrix products.
\end{remark}
 The $^*$-algebra $\mathcal{A}_0$ can be mapped onto a $^*$-algebra of bounded linear operators acting on the space of quantum states $\mathfrak{H}_\omega = \ell^2(\Xi^\omega)$, see~\eqref{def:QuantumStates}, via the representation formula:
	\begin{equation}\label{eq:representationformula}
	\pi_\omega(f) \phi (x) = \sum_{(\omega, y) \in \Gamma^{(\omega)}} f \left (\mathfrak{t}_{-x} \omega, y - x \right ) \phi(y), \qquad \forall \phi \in \mathfrak{H}_\omega, \quad \forall x \in \Xi^\omega.
	\end{equation}
	In the above formula, we have implicitly used the fact that the set $\Xi^\omega$ defined by \eqref{eq:Xiomega} coincides with $\left\{ y \in \G_D \; \vert \; (\omega,y) \in \Gamma^{(\omega)}\right\}$. 

\begin{figure}[h]
\centering
\begin{minipage}{.49\textwidth}\centering
	\includegraphics[height=3.5cm]{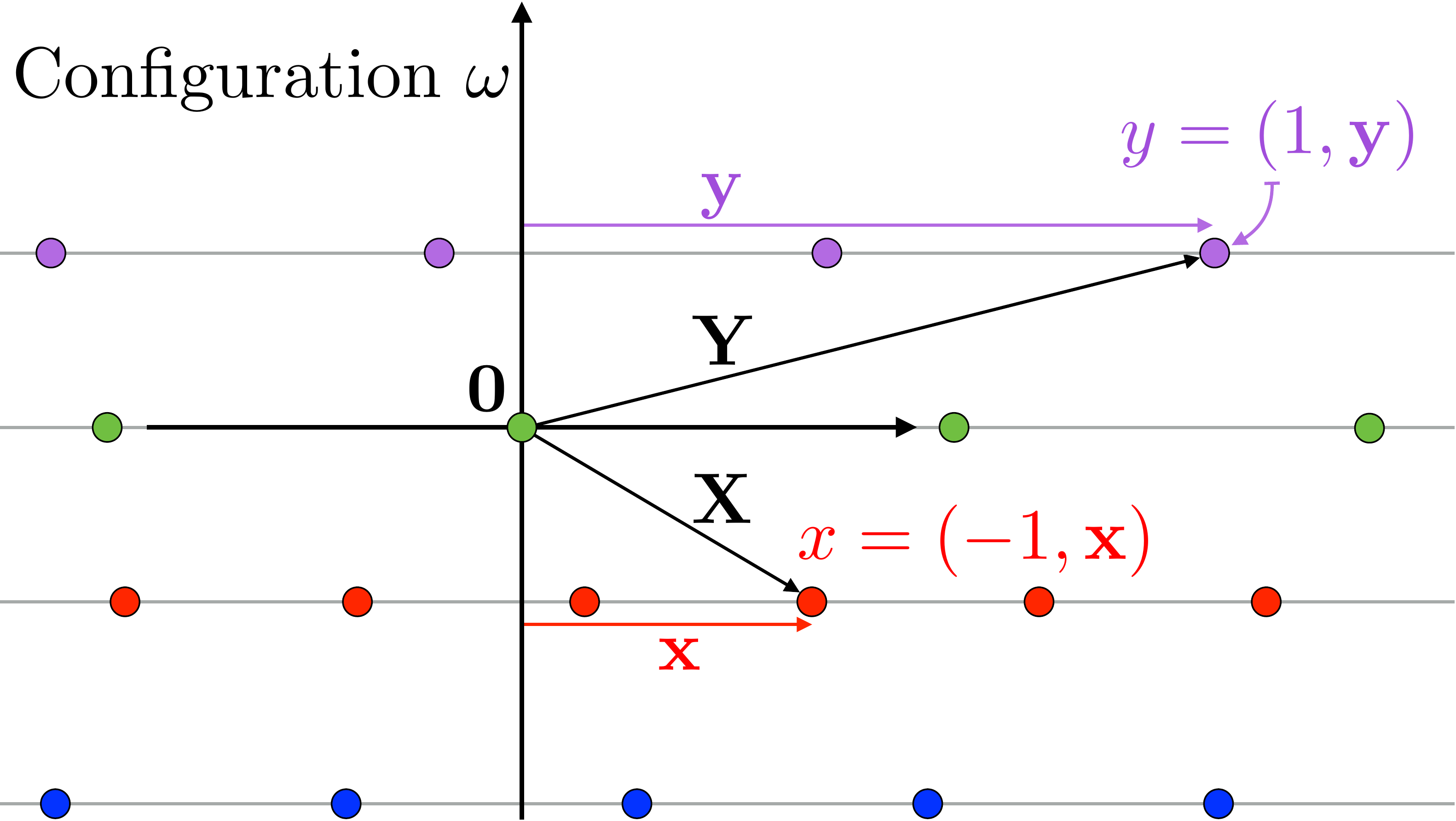}
\end{minipage}
\begin{minipage}{.49\textwidth}\centering
	\includegraphics[height=3.5cm]{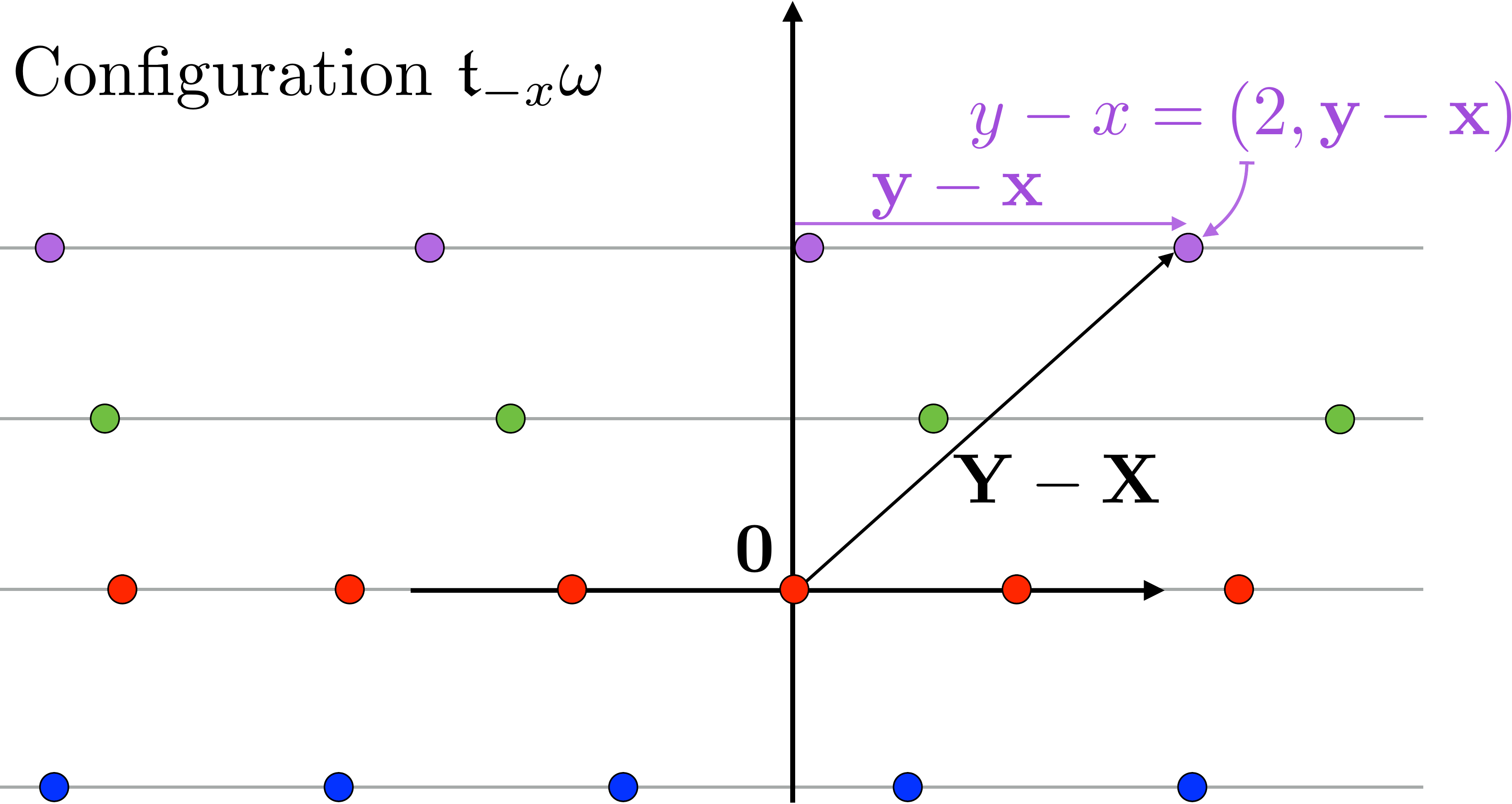}
\end{minipage}
\caption{Graphical explanation of the representation formula \eqref{eq:representationformula} for $d=1$, $p=4$, and for a configuration $\omega=(3,\bgamma_1,\bgamma_2,{\mathbf 0},\bgamma_4) \in \Omega_D$. The points $x=(2-3,\bx)=(-1,\bx)$ and $y=(4-3,\by)=(1,\by)$ of $\Xi^\omega$ are combined to form an element $y-x=(2,\by-\bx)$ such that $(\mathfrak{t}_{-x}\omega,y-x) \in \Gamma^{({\mathfrak{t}_{-x}}\omega)}$.}
\label{fig:figure4}
\end{figure}

	The following covariance condition holds:	for a given $\gamma = (\omega, a) \in \Gamma(X)$,
	\begin{equation}
	 \pi_{\mathfrak{t}_{-a}\omega}(f)=\mathfrak{T}(\gamma)^{-1}\pi_\omega(f) \mathfrak{T}(\gamma)  ,
	\end{equation}
	where the translations $\mathfrak{T}(\gamma): \mathfrak{H}_{\mathfrak{t}_{-\ba}\omega} \to \mathfrak{H}_\omega$ are unitary operators defined by
	\begin{equation}
		\mathfrak{T}(\gamma)\phi(x) = \phi(\mathfrak{t}_{-a}x), \qquad \forall \phi \in \mathfrak{H}_{\mathfrak{t}_{-a}\omega}, \quad \forall x \in \Xi^\omega.
	\end{equation}
\begin{remark}
	When $d = 2$, magnetic fields can be included in the description through a Peierls substitution term as follows. We assume that there is only one orbital per unit cell, and that the orbitals are carried by atoms located at the lattice sites. Let $\mathscr B$ be an antisymmetric $3 \times 3$ matrix representing the magnetic field.
	For $x=(l,\bx)$ and $y= (m,\by)$ in $\Xi^{\omega}$ with $\omega = (j,\bgamma_1, \dots, \bgamma_p)$, we define the corresponding Peierls phase term as
	\begin{equation}
		\Phi^\omega_{\mathscr B}(x, y) = \pi \sum_{\nu, \mu=1}^3 \mathrm{\mathscr  B}_{\nu, \mu} X_\nu Y_\mu,
	\end{equation}
	where $\mathbf{X} = \bx +( h_{j+l} - h_j)\be_{3}$ and $\mathbf{Y} = \by+( h_{j+m} - h_j)\be_{3}$ are the physical space vectors of the jumps from the origin to the sites $x$ and $y$.

	  The product and representation formulae are then modified as
	\begin{equation}\label{def:magnetic}
		\begin{aligned}
			(f \starb g)(\omega, a) &= \sum_{(\omega,b) \in \Gamma^{(\omega)}} f(\omega, b) g(\mathfrak{t}_{-b} \omega, a - b) e^{i \Phi^\omega_{\mathscr B}(a, b)},\\
			\pi_\omega(f) \phi (x) &= \sum_{(\omega,y) \in \Gamma^{(\omega)}} f \left (\mathfrak{t}_{-x} \omega, y - x \right ) e^{-i\Phi^\omega_{\mathscr B}(x, y)} \phi(y), \qquad \forall \phi \in \mathfrak{H}_\omega, x \in \Xi^{\omega},
		\end{aligned}
	\end{equation}
	where we again used the fact that $\Xi^\omega = \left\{ y \in \G_D \; \vert \; (\omega,y) \in \Gamma^{(\omega)}\right\}$.
	In this case, the following covariance condition holds:
	for a given $\gamma = (\omega, a) \in \Gamma(X)$,
	\begin{equation}
	\pi_{\mathfrak{t}_{-a}\omega}(f) = U(\gamma)^{-1} \pi_\omega(f) U(\gamma)    ,
	\end{equation}
	where the magnetic translations $U(\gamma): \mathfrak{H}_{\mathfrak{t}_{-a}\omega} \to \mathfrak{H}_\omega$ are unitary operators defined by
	\begin{equation}
	U(\gamma)\phi(x) = \mathrm{exp} \left ( i \int_{[\mathbf{X} - \mathbf{A}, \mathbf{X}]} \mathscr{A}(y) \cdot dy \right ) \phi(\mathfrak{t}_{-a}x), \qquad \forall \phi \in \mathfrak{H}_{\mathfrak{t}_{-a}\omega}, \quad \forall x \in \Xi^\omega.
	\end{equation}
	Here, $\mathscr{A}$ is a vector potential giving rise to the magnetic field $\mathscr  B$, and $[\mathbf{X} - \mathbf{A}, \mathbf{X}]$ is the line segment joining $\mathbf{X} - \mathbf{A}$ to $\mathbf{X}$ in $\R^3$.
\end{remark}
Note that $\pi_\omega(f)$ is hermitian if $f = f^*$.
This representation induces a $C^*$ norm on $\mathcal{A}_0$, defined by:
\[
\Vert f \Vert = \sup_{\omega \in X} \Vert \pi_\omega(f) \Vert,
\]
where the norm on the right hand side is the operator norm on $\mathcal{L}(\mathfrak{H}_\omega)$, the space of bounded linear operators on $\mathfrak{H}_\omega$. We then construct in a third and final step the $C^*$-algebra $\mathcal{A}$ as the completion of $\mathcal{A}_0$ for this norm.

\paragraph{Integro-differential calculus}
Integration and derivations can be constructed on the $C^*$-algebra as follows.
Let $\Prob$ be a probability measure on $X$ invariant by $\Gamma(X)$-action.
A formal integration is then obtained on $\mathcal{A}_0$ as follows: if $f \in \mathcal{A}_0$,
\begin{equation}\label{eq:trace}
\mathcal{T}_\Prob(f) = \int_X d \Prob(\omega) f(\omega, \bzero).
\end{equation}
This trace is positive, $\mathcal{T}_\Prob(f^* *f) \geq 0$. It is faithful ($\mathcal{T}_\Prob(f^* *f) = 0$ iff $f = 0$) when the support of the measure $\Prob$ is $X$.
\begin{remark}
	This trace coincides with the trace per unit volume in $\R^d$ of the corresponding operator $\pi_\omega(f)$ thanks to the Birkhoff property~\eqref{def:continuousBirkhoff}.
\end{remark}
One then defines $L^p$-norms by setting for $1 \leq p < \infty$,
\begin{equation}
	\Vert f \Vert_p = \Trace{ (f * f^*)^{p/2}}^{1/p} \qquad \forall f \in \mathcal{A}_0.
\end{equation}
Closure of $\mathcal{A}_0$ with respect to the $L^p$-norms defines Banach spaces $L^p(\mathcal{A}, \mathcal{T}_\Prob)$.
Derivations are introduced next by setting
\begin{equation}\label{eq:derivation}
	\partial_j f(\omega, x) = i x_j f(\omega, x), \qquad f \in \mathcal{A}_0.
\end{equation}

Given a multi-index $\boldsymbol{\alpha} = (\alpha_1, \dots, \alpha_d)$ we will use the notation $\partial_{\boldsymbol{\alpha}} f$ for $\partial_1^{\alpha_1} \dots \partial_d^{\alpha_d}$. The derivations defined above satisfy the fundamental properties of derivation operators: they commute, they are $^*$-derivations in the sense that $\partial_j (f^*) = (\partial_j f)^*$, and satisfy the Leibniz rule $\partial_j(f*g)=\partial_jf * g+ f * \partial_jg$ . Finally one has the operator representation:
\begin{equation}
	\pi_\omega(\partial_j f) = -i \left [ x^\omega_j, \pi_\omega f \right ],
\end{equation}
where $\mathbf{x}^\omega = (x^\omega_1, \dots, x^\omega_d)$ is the position operator on $\cL^\omega$.
\begin{remark}
	In the periodic case, these derivations correspond (by Fourier transform) to derivation in quasi-momentum space.
\end{remark}

\paragraph{Analytic spectral calculus}
A basic notion in operator algebras is that of the resolvent sets, the spectrum, and the introduction of a spectral calculus by a complex contour integral. Since the operators we are considering here are bounded, this is straightforward and we refer e.g. to~\cite{Arveson2006} for the details. For a given element $f \in \mathcal{A}$, its resolvent set $\rho(f)$ and spectrum $\sigma(f)$ are
\[
	\rho(f) = \{ z \in \C \quad \text{s.t.} \quad ( z \mathbf{1} - f ) \text{ is invertible} \}, \qquad \sigma(f) = \C \setminus \rho(f).
\]
Furthermore, the resolvent set is open in $\C$, and the resolvent function $z \mapsto (z \mathbf{1} - f)^{-1}$ is an algebra-valued analytic function of $z \in \rho(f)$. An analytical calculus can then be defined on $\mathcal{A}$ as follows. Let $\mathcal{F}_{\sigma(f)}$ by the algebra of $\C$-valued functions which are analytic in an open neighborhood of $\sigma(f)$. Then we define a homomorphism of algebras
\begin{equation}\label{def:AnalyticalCalculus}
	\mathcal{F}_{\sigma(f)} \owns \phi \mapsto \frac{1}{2i\pi} \oint_\mathcal{C} \phi(z) (z \mathbf{1} - f)^{-1} \mathrm{d}z \in \mathcal{A},
\end{equation}
where $\mathcal{C}$ is a contour surrounding $\sigma(f)$ in the analyticity domain of $\phi$. The right-hand side of this mapping is independent of the particular choice of contour, and usually noted as $\phi(f)$.

\subsection{Perfect bilayers}\label{sec:bilayers}
Let us identify further the groupoid and the associated $C^*$-algebra in the specific geometry of a bilayer system with no disorder.

\paragraph{Transversal and groupoid.}
We start by identifying the transversal in the case of the hull $\Omega_D$ of a perfect incommensurate bilayer. We have
\[
\Omega_D = \left\{1,2\right\} \times \Gamma_1 \times \Gamma_2,
\]
and
\[
X = \left \{ (1,\bzero, \bgamma_2)\ ; \; \bgamma_2 \in \Gamma_2 \right \} \cup \left \{ ( 2,\bgamma_1, \bzero)\ ; \; \bgamma_1 \in \Gamma_1 \right \}.
\]
This leads to the identification:
\begin{equation}\label{def:transversal2}
X \equiv X_1 \cup X_2 \qquad \text{ where } \qquad \begin{cases}
X_1 = \Gamma_2,\\
X_2 =  \Gamma_1.
\end{cases}
\end{equation}
The set $X_1$ (resp. $X_2$) describes the possible configurations from the point of view of a given lattice site of layer $\cL_1$ (resp. $\cL_2$). In a configuration $\bgamma_2 \in X_1$, which corresponds to $\omega=(1,\bzero, \bgamma_2)$, the lattice sites of $\cL_1$ are located at $\cR_1$ while the lattice sites of $\cL_2$ are located at $\bgamma_2+\cR_2+h\be_{d+1}$, where $h$ is the distance between the two layers. In a configuration $\bgamma_1 \in X_2$, which corresponds to $\omega=(2, \bgamma_1,\bzero)$, the lattice sites of $\cL_2$ are located at $\cR_2$ while the lattice sites of $\cL_1$ are located at $\bgamma_1+\cR_1-h\be_{d+1}$.

The decomposition~\eqref{def:transversal2} of the transversal allows us to identify a block decomposition of the groupoid:
\begin{equation}\label{def:groupoid2}
\Gamma(X) \equiv \tGamma_{11} \cup \tGamma_{12} \cup \tGamma_{21} \cup \tGamma_{22} = \tGamma(X_1, X_2),
\end{equation}
where the set of arrows $\tGamma_{11}$ and $\tGamma_{22}$ include all intralayer jumps:
\begin{equation*}
	\begin{cases}
		\tGamma_{11} &= \left \{ \left (\bgamma_2,  \bm \right) ; \;  \bgamma_2 \in X_1, \;  \bm \in \cR_1 \right \},\\
		\tGamma_{22} &= \left \{ \left (\bgamma_1,  \bn \right) ; \;  \bgamma_1 \in X_2, \;  \bn \in \cR_2 \right \},
	\end{cases}
\end{equation*}
while $\tGamma_{21}$ describes jumps from the first to the second layer, and $\tGamma_{12}$ jumps from the second to the first layer.
Let us note that the above representation would be redundant in the case of interlayer jumps, e.g., for $\left (\bgamma_2,  \bq \right)$ in $\tGamma_{12}$. Indeed, $\bgamma_2 \in X_1$ can be deduced from $\bq \in \bgamma_2 + \cR_2$ as its equivalence class modulo $\cR_2$. We will thus denote these arrows solely as arrow vectors $\tq$, observing that these can take their values in all of $\R^d$:
\begin{equation}
	\begin{cases}
		\tGamma_{12} &= \left \{ \tq ; \; \bq \in \R^d \right \}, \\
		\tGamma_{21} &= \left \{ \tp ; \; \bp \in \R^d \right \}.
	\end{cases}
\end{equation}
This allows us to make explicit the three groupoid operations, namely the range map
\begin{equation}\label{def:range2}
r: \tGamma \to \tGamma, \quad
\left \{ \begin{aligned}
 \tGamma_{11} \to X_1, && (\bgamma_2, \bm) &\mapsto \bgamma_2, \\
 \tGamma_{12} \to X_1, && \tq &\mapsto \class2{\bq}, \\
 \tGamma_{21} \to X_2, && \tp &\mapsto \class1{\bp}, \\
 \tGamma_{22} \to X_2, && (\bgamma_1, \bn) &\mapsto \bgamma_1,
\end{aligned} \right.,
\end{equation}
the source map:
\begin{equation}\label{def:source2}
s: \tGamma \to \tGamma, \quad
\left \{ \begin{aligned}
 \tGamma_{11} \to X_1, && (\bgamma_2, \bm) &\mapsto \bgamma_2 - \class2{\bm}, \\
 \tGamma_{12} \to X_2, && \tq &\mapsto - \class1{\bq}, \\
 \tGamma_{21} \to X_1, && \tp &\mapsto  - \class2{\bp}, \\
 \tGamma_{22} \to X_2, && (\bgamma_1, \bn) &\mapsto \bgamma_1 -  \class1{\bn},
\end{aligned} \right.
\end{equation}
and the composition map:
\begin{equation}\label{def:composition2}
\left \{ \begin{aligned}
\text{For } (\bgamma_2, \bm) \in \tGamma_{11}, && \bm' \in \cR_1:
	&& (\bgamma_2, \bm) &\circ (\bgamma_2 - \class2{\bm}, \bm') &&= (\bgamma_2, \bm + \bm') &&\in \tGamma_{11},\\
\text{for }  \tq \in \tGamma_{12}, && \bp + \bq\in \cR_1:
 	&& \tq &\circ \tp &&= (\class2{\bq}, \bq + \bp)  &&\in \tGamma_{11}, \\
\text{for }  \tp \in \tGamma_{21}, && \bq  + \bp \in \cR_2:
	&& \tp &\circ \tq &&= (\class1{\bp}, \bp + \bq)  &&\in \tGamma_{22}, \\
\text{for }  (\bgamma_1, \bn) \in \tGamma_{22}, && \bn' \in \cR_2:
 	&& (\bgamma_1, \bn) &\circ  (\bgamma_1 - \class1{\bn}, \bn')  &&= (\bgamma_1, \bn + \bn')  &&\in \tGamma_{22},\\
\text{for }  \tp \in \tGamma_{21}, && \bm \in \cR_1:
	&& \tp &\circ (-\class2{\bp}, \bm) &&= \overrightarrow{\bp + \bm}  &&\in \tGamma_{21},\\
\text{for }  (\bgamma_1, \bn) \in \tGamma_{22}, && \bp + \bn \in \bgamma_1 + \cR_1:
 	&& (\bgamma_1, \bn) &\circ \tp &&= \overrightarrow{\bn + \bp}  &&\in \tGamma_{21},\\
\text{for } (\bgamma_2, \bn) \in \tGamma_{11}, && \bq + \bn \in \bgamma_2 + \cR_2:
	&& (\bgamma_2, \bm) &\circ \tq &&= \overrightarrow{\bm + \bq} &&\in \tGamma_{12},\\
\text{for }  \tq \in \tGamma_{12}, && \bn \in \cR_2:
 	&& \tq &\circ ( - \class1{\bq}, \bn) &&= \overrightarrow{\bq + \bn}  &&\in \tGamma_{12},
\end{aligned} \right.
\end{equation}
where we have introduced the notation $\class{j}{\br}=\br+\cR_j$, $j = 1 \dots p$ to denote the equivalence classes in $\Gamma_j$ of a given point $\br\in \R^d$.
Finally, the fiber $\Gamma^{(\omega)}$ can be determined:
\begin{subequations}
\begin{align}
\tGamma^{(\bgamma_2)} = \Gamma^{(1,\bzero, \bgamma_2)} \equiv \cL^{(1,\bzero, \bgamma_2)}, \quad \text{for } \bgamma_2 \in X_1,\\
\tGamma^{(\bgamma_1)} = \Gamma^{(2,\bgamma_1, \bzero)} \equiv \cL^{(2,\bgamma_1, \bzero)}, \quad \text{for } \bgamma_1 \in X_2.
\end{align}
\end{subequations}
where $\cL^\omega$ is the set of points defined by~\eqref{def:parameterizationatomicset}.

We obtain for the groupoid the discrete counterpart to the ergodic property of the continuous hull, Proposition~\ref{prop:continuousHullErgodicity}. Let $\mathrm{d}\bgamma$ denote the usual Lebesgue measure on $X_1 = \Gamma_2$ and $X_2 = \Gamma_1$.
\begin{proposition}\label{prop:discreteHullErgodicity}
	Let $\Prob$ be the probability measure on $X$ with uniform density $(\vert \Gamma_1 \vert + \vert \Gamma_2 \vert)^{-1} \mathrm{d}\bgamma$.
	\begin{enumerate}
	\item $\Prob$ is invariant by the groupoid action.
	\item The dynamical system $(X,\Gamma(X), \mathfrak{t}, \Prob)$ is uniquely ergodic if and only if the lattices $\cR_1, \cdots , \cR_p$ are incommensurate. In this case, we have the Birkhoff property: for any $f \in C(X)$ and $\omega \in X$,
	\begin{equation}\label{def:discreteBirkhoff}
		\lim_{r \to \infty} \frac{1}{\# \left ( B_r \cap \cL^\omega \right )} \sum_{\ba \in B_r \cap \cL^\omega } f(\mathtt{t}_{-\ba} \omega) 
= \int_X f \mathrm{d}\Prob,
	\end{equation}
	where $B_r$ is the ball of radius $r$ centered at the origin.
\end{enumerate}
\end{proposition}
The proof of this result follows the same lines as the proof of Proposition~\ref{prop:continuousHullErgodicity}. We therefore omit it for the sake of brevity.

\paragraph{Bilayer $C^*$-algebra.}
For the case of periodic incommensurate bilayers with one orbital per unit cell, we can give a comprehensive description of the abstract algebra defined above. First, elements of $\mathfrak{H}_\omega = \ell^2(\Xi^\omega)$ can be seen as:
\begin{itemize}
\item elements of $\mathfrak{H}_{\bgamma_2} = \ell^2(\cR_1) \oplus \ell^2(\bgamma_2 + \cR_2)$ for $\omega \equiv \bgamma_2 \in X_1$,
\item or elements of $\mathfrak{H}_{\bgamma_1} = \ell^2(\bgamma_1+\cR_1) \oplus \ell^2(\cR_2)$ for $\omega \equiv \bgamma_1 \in X_2$.
\end{itemize}
Then, given the decomposition \eqref{def:groupoid2} of $\Gamma(X)$, it makes sense to write a block decomposition of functions $f \in C^*(\Gamma(X))$ as:
\begin{equation}\label{def:BlockDecomposition}
f = \begin{bmatrix}
f_{11} & f_{12} \\ f_{21} & f_{22}
\end{bmatrix}
\end{equation}
where
\begin{equation}\label{def:BlockCoefficients}
\begin{aligned}
f_{11}: \quad & \tGamma_{11} = \Gamma_2 \times \cR_1 \to \C, \\
f_{12}: \quad & \tGamma_{12} \equiv \R^d \to \C, \\
f_{21}: \quad & \tGamma_{21} \equiv \R^d \to \C, \\
f_{22}: \quad & \tGamma_{22} = \Gamma_1 \times \cR_2 \to \C.
\end{aligned}
\end{equation}
Note that the decomposition~\eqref{def:BlockDecomposition} of the tight-binding hopping parameters into intra- and inter-layer terms is naturally the parameterization used in physics~\cite{Fang2015}. In particular, inter-layer coefficients are usually represented directly as a continuous function of the relative position of the atoms as in~\eqref{def:BlockCoefficients}, see e.g.,~\cite{Fang2016}.
\begin{subequations}
Let us now write the $*$ product defining the algebra:
\begin{itemize}
\item for $\bgamma_2 \in X_1$ and $\bm \in \cR_1$,
\begin{equation}
\begin{aligned}
(f * g)_{11}(\bgamma_2, \bm) = & \sum_{\bm' \in \cR_1} f_{11}(\bgamma_2, \bm') g_{11}(\bgamma_2 - \class2{\bm'}, \bm - \bm') \\ &+  \sum_{\bq' \in \bgamma_2 + \cR_2} f_{12}(\bq') g_{21}(\bm - \bq');
\end{aligned}
\end{equation}
\item for $\bq \in \R^2$,
\begin{equation}
\begin{aligned}
(f * g)_{12}(\bq) = & \sum_{\bm' \in \cR_1} f_{11}(\bgamma_2, \bm') g_{12}(\bq - \bm')  \\ &+  \sum_{\bn' \in \cR_2} f_{12}(\bq - \bn') g_{22}(\class1{\bn' - \bq}, \bn');
\end{aligned}
\end{equation}
\item for $\bp \in \R^2$,
\begin{equation}
\begin{aligned}
(f * g)_{21}(\bp) = & \sum_{\bn' \in \cR_2} f_{22}(\class1{\bp}, \bn') g_{21}(\bp - \bn')\\
& +  \sum_{\bm' \in \cR_1} f_{21}(\bp - \bm') g_{11}(\class2{\bm' - \bp}, -\bm');
\end{aligned}
\end{equation}
\item for $\bgamma_1 \in X_2$ and $\bn \in \cR_2$,
\begin{equation}
\begin{aligned}
(f * g)_{22}(\bgamma_1, \bn) = & \sum_{\bn' \in \cR_2} f_{22}(\bgamma_1, \bn') g_{22}(\bgamma_1 - \class1{\bn'}, \bn - \bn') \\
& +  \sum_{\bp' \in \bgamma_1 + \cR_1} f_{21}(\bp') g_{12}(\bn - \bp').
\end{aligned}
\end{equation}
\end{itemize}
\end{subequations}
We can also write the $^*$ operation as follows: for $\bgamma_2 \in X_1$, $\bgamma_1 \in X_2$ and $\bm \in \cR_1$, $\bq \in \bgamma_2$, $\bp \in \bgamma_1$, $\bn \in \cR_2$:
\begin{equation}
\begin{aligned}
(f^*)_{11}(\bgamma_2, \bm) &= \overline{f_{11}(\bgamma_2 - \class2{\bm}, -\bm)}, & (f^*)_{12}(\bq) &= \overline{f_{21}(-\bq)}, \\
(f^*)_{21}(\bp) &= \overline{f_{12}(- \bp)}, \hspace{1cm} & (f^*)_{22}(\bgamma_1, \bn) &= \overline{f_{22}(\bgamma_1 - \class1{\bn}, -\bn)}.
\end{aligned}
\end{equation}
Finally, the representation formula writes as follows:
\begin{enumerate}
\item for $\bgamma_2 \in X_1$ and $\phi  = (\phi_1, \phi_2)\in  \ell^2(\cR_1) \oplus \ell^2(\bgamma_2 + \cR_2)$, $\pi_{\bgamma_2} (f) \phi$ can be decomposed in $\ell^2(\cR_1) \oplus \ell^2(\bgamma_2 + \cR_2)$ as
\begin{equation}\label{eq:RepresentationFormulaBilayer}
\left \{
\begin{aligned}
\left ( \pi_{\bgamma_2} (f) \phi\right )_1 (\bm) &= \sum_{\bm' \in \cR_1} f_{11} \left (\bgamma_2 - \class2{\bm}, \bm' - \bm \right ) \phi(\bm') + \sum_{\bq' \in \bgamma_2 + \cR_2} f_{12} \left (\bq' - \bm \right )  \phi(\bq'), \\
\left ( \pi_{\bgamma_2} (f) \phi\right )_2 (\bq) &= \sum_{\bm' \in \cR_1} f_{21} \left (\bm' - \bq \right )  \phi(\bm') + \sum_{\bn' \in \cR_2} f_{22} \left (- \class1{\bq}, \bn' \right )  \phi(\bq+\bn'), \\
&  \text{ for all }  \bm \in \cR_1 \text{ and } \bq \in \bgamma_2 + \cR_2;
\end{aligned}
\right.
\end{equation}
\item for $\bgamma_1 \in X_2$ and $\phi  = (\phi_1, \phi_2)\in  \ell^2(\bgamma_1 + \cR_1) \oplus \ell^2(\cR_2)$, $\pi_{\bgamma_1} (f) \phi$ can be decomposed in $\ell^2(\bgamma_1 + \cR_1) \oplus \ell^2(\cR_2)$ as
\begin{equation}
\left \{
\begin{aligned}
\left ( \pi_{\bgamma_1} (f) \phi\right )_1 (\bp) &= \sum_{\bm' \in \cR_1} f_{11} \left (-\class2{\bp}, \bm' \right ) \phi(\bp+\bm') + \sum_{\bn' \in \cR_2} f_{12} \left (\bn' - \bp \right ) \phi(\bn'), \\
\left ( \pi_{\bgamma_1} (f) \phi\right )_2 (\bn) &= \sum_{\bp' \in \bgamma_1+\cR_1} f _{21}\left (\bp' - \bn \right ) \phi(\bp') + \sum_{\bn' \in \cR_2} f_{22} \left (\bgamma_1 - \class1{\bn}, \bn' - \bn \right ) \phi(\bn'), \\
& \text{ for all } \bp \in \bgamma_1 + \cR_1 \text{ and } \bn \in \cR_2.
\end{aligned}
\right.
\end{equation}
\end{enumerate}

\paragraph{Trace per unit volume and integro-differential calculus.}

A classical consequence~\cite{bellissard1994noncommutative} of the ergodicity of $\Prob$ on $X$ under the action of $\Gamma(X)$, Proposition~\ref{prop:discreteHullErgodicity}, is that we can characterize uniquely $\Prob$ under the condition that it is a trace per unit volume:
\begin{proposition}
When $\cR_1$ and $\cR_2$ are incommensurate, the invariant, ergodic probability measure $\Prob$ is uniquely defined as a trace per unit volume in the sense that it satisfies: for any $f \in C^*(\Gamma(X))$ and $\omega \in X$,
\begin{equation}\label{eq:traceergodicity}
\mathcal{T}_\Prob(f) = \lim_{r \to \infty} \frac{1}{\# \left ( B_r \cap \cL^\omega \right ) } \mathrm{Tr} \left (\pi_\omega(f) \vert_{B_r} \right ).
\end{equation}
\end{proposition}
Moreover, for $j=1,\cdots,d,$ the derivation defined by~\eqref{eq:derivation} extends easily: for $\bgamma_2 \in X_1$ and $\bn \in \cR_1$, $\bq \in \bgamma_2$, $\bp \in \bgamma_1$, $\bm \in \cR_2$,
\begin{subequations}
\begin{equation}
\begin{aligned}
(\partial_j f)_{11}(\bgamma_2, \bn) &= i \mathrm{n}_j f_{11}(\bgamma_2, \bn), & (\partial_j f)_{12}(\bq) &= i \mathrm{q}_j f_{12}(\bq), \\
(\partial_j f)_{21}(\bp) &= i \mathrm{p}_j f_{21}(\bp), \hspace{1cm} & (\partial_j f)_{22}(\bgamma_1, \bm) &= i \mathrm{m}_j f_{22}(\bgamma_1, \bm) .
\end{aligned}
\end{equation}
There exists also a derivation for $j=d+1$
\begin{equation}
\begin{aligned}
(\partial_{d+1} f)_{11}(\bgamma_2, \bn) &= 0, & (\partial_{d+1}  f)_{12}(\bq) &= +i f_{12}(\bq), \\
(\partial_{d+1}  f)_{21}(\bp) &= -i f_{21}(\bp), \hspace{1cm} & (\partial_{d+1}  f)_{22}(\bgamma_1, \bm) &= 0.
\end{aligned}
\end{equation}
\end{subequations}

Note that the derivation operator $\partial_{d+1}$ is bounded on $\mathcal{A}$. Hence, an element $f$ belongs to $\mathcal{C}^N$ if and only if $\Vert \partial_{\boldsymbol{\alpha}} f \Vert < \infty$ for all multi-indexes $\boldsymbol{\alpha}$ such that $\sum_{j=1}^d \alpha_j \leq N$.

\begin{remark}
	The extension to multilayer systems of the formalism introduced in the previous sections is straightforward, though somewhat cumbersome.
\end{remark}

\subsection{Noncommutative Kubo formula}

Electronic transport can be modeled as the result of the interplay between the quantum evolution in the presence of a uniform electric field and dissipation mechanisms such as scattering events which depend on the environment of the charge carrier (electron or hole) with charge $q$. As discussed in the lecture notes~\cite{bellissard2003coherent}, models for these dissipation mechanisms can be obtained by considering microscopic, many-body electron systems coupled with an environment such as a phonon bath, and then integrating out the degrees of freedom of the environment to obtain an effective single-particle model. As in the Drude model, dissipation mechanisms are represented by discrete scattering events (collisions) with Poisson-distributed independent time delays between successive collisions.

Given any one-particle density matrix $\rho$ as an initial state (i.e. a positive element of the $C^*$-algebra associated with a quantum system), the effective, collision-averaged one-particle time evolution is given by the Liouville equation~\cite{schulz1998kinetic}
\begin{equation}
	\frac{\mathrm{d}\rho}{\mathrm{d}t} + \mathcal{L}_{h - q \mathcal{E}(t) \cdot \vec{x}}(\rho) = - \frac{\boldsymbol{1} - \widehat{\kappa}^*}{\tau}(\rho),
\end{equation}
where
\begin{itemize}
	\item $h$ is the effective single-electron tight-binding Hamiltonian,
	\item $\mathcal{E}$ is a constant or time-harmonic spatially uniform electric field with frequency $\widehat{\omega}$,
	\item $\vec{x}$ is the position operator and $\boldsymbol{1}$ is the identity operator,
	\item $q$ is the charge of the carrier (hole or electron),
	\item $\mathcal{L}_{h - q \mathcal{E}(t) \cdot \vec{x}} = i/\hbar \left [ h, \cdot \right ] - \frac{q}{\hbar} \mathcal{E}(t) \cdot \nabla $ is the Liouvillian operator governing the one-particle time evolution in the absence of collisions,
	\item $\widehat{\kappa}^*$ is a scattering-event-averaged {\em collision efficiency operator},
	\item $\tau$ is the mean collision time of a Poisson process with law $e^{-t/\tau} \mathrm{d}t/\tau$ governing the time delay between the independent scattering events (collisions).
\end{itemize}

The linear conductivity tensor is then readily available in the framework of $C^*$-algebras by the famous noncommutative Kubo formula~\cite{schulz1998kinetic}:
\begin{equation}\label{eq:KuboFormula}
	\sigma_{ij}(\widehat{\omega}) = \frac{q^2}{\hbar^2} \Trace{ \partial_i h \left [ (\boldsymbol{1} - \widehat{\kappa}^*) / \tau + \mathcal{L}_h - i \widehat{\omega} \right ]^{-1} \partial_j f_{\beta, \mu}(h)},
\end{equation}
where $f_{\beta, \mu}(h) = \frac{1}{1 + e^{\beta (h - \mu)}}$ is the Fermi-Dirac one-particle density matrix in the grand-canonical equilibrium with chemical potential (Fermi level) $\mu$ and at temperature $T$ such that $\beta = 1/k_B T$.

In the relaxation time approximation (RTA), the effective relaxation operator $(\boldsymbol{1} - \widehat{\kappa}^*) / \tau$ is replaced by $\boldsymbol{1}/\tau_\mathrm{rel}$, a single relaxation time which depends in general strongly on the temperature.
Full details of the modeling assumption and derivation of~\eqref{eq:KuboFormula} can be found in the references~\cite{bellissard1994noncommutative,schulz1998kinetic,bellissard2003coherent,prodan2012quantum}.

\section{Numerical example} \label{sec:numerics}
We propose in this section a minimalistic one-dimensional toy model to study incommensurability effects in multilayer systems, which is new up to our knowledge.
We then present a numerical strategy and numerical results for the computation of the density of states and the conductivity. In this paper, we will restrict ourselves to scanning rational values of the lattice ratio parameter. We therefore use the traditional approach of constructing periodic supercells~\cite{prodan2012quantum}, and we never actually compute directly quantities of interest in the case of incommensurate lattice ratios.
In future publications, we will propose strategies based directly on the $C^*$-algebra representation, possibly addressing directly the incommensurate case, and with rigorous error control. This substantial additional effort is currently ongoing.

\subsection{Description of the model}

Let us consider two parallel one-dimensional crystals with lattice constants $\ell_1$ and $\ell_2$ normalized such that
\begin{equation}\label{eq:normalization1d}
	\ell_1 \ell_2 = 1.
\end{equation}
Following the notation introduced in Sections~\ref{sec:hull} and~\ref{sec:CSA}, we set the lattices $\cR_j = \ell_j \Z$ and the unit cells $\Gamma_j = \R / \cR_j \equiv \ell_j \T$  for $j = 1,2$, where $\T := \R / \Z$ is the one-dimensional $1$-periodic torus. The set of all possible configurations is parameterized by the hull $\Omega = \Gamma_1 \times \Gamma_2$.
We then form a quantum lattice model with one orbital per unit cell in each layer. The relevant parameterization is provided by the transversal $X = X_1 \cup X_2$ with $X_1 \equiv \ell_2 \T$ and $X_2 \equiv \ell_1 \T$.

\begin{figure}[hb!]
	\centering
	\includegraphics[width=.6\textwidth]{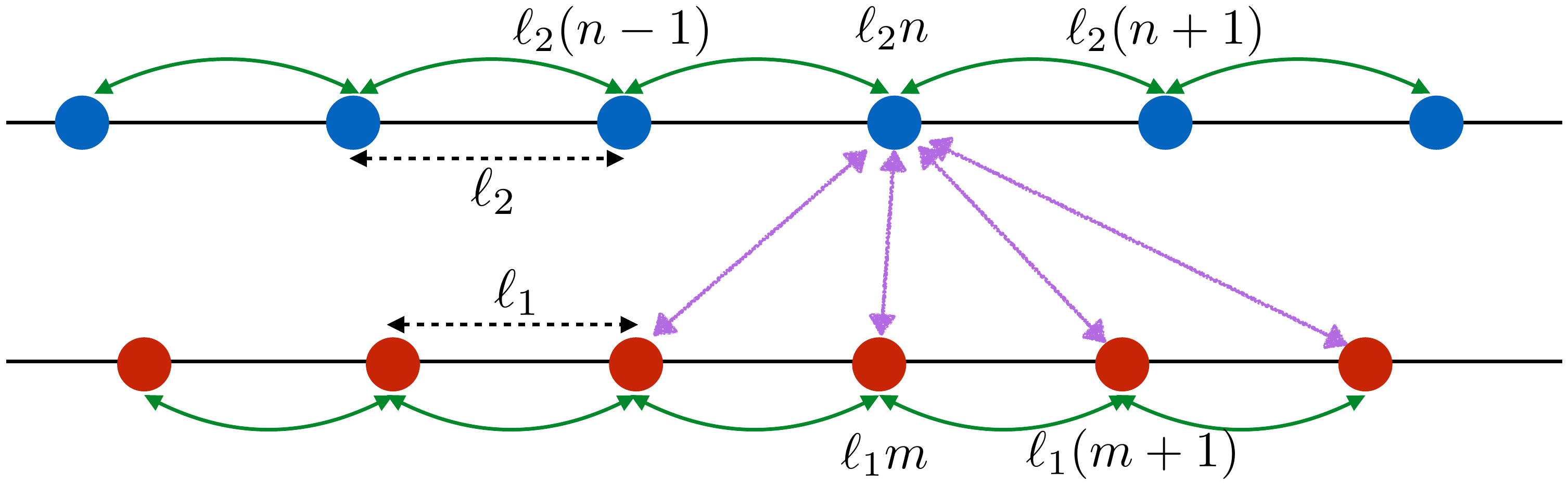}
	\caption{Sketch of the geometry and tight-binding hoppings of our 1D toy model.}
	\label{fig:ToyModel}
\end{figure}

We consider the model Hamiltonian generated by an element $h$ of the $C^*$-algebra presented in Section~\ref{sec:bilayers} and which takes the values
\begin{equation}\label{def:ModelHamiltonian}
	\begin{aligned}
	&h_{11}(\bgamma_2,  m) = \begin{cases} 1 & \text{if } \vert m \vert = \ell_1, \\ 0 & \text{otherwise,} \end{cases} \hspace{2cm} & \text{for } \bgamma_2 \in \ell_2 \T,\quad m \in \ell_1\Z,\\
	&h_{12}(p) = W e^{-\frac{1}{2} \left ( \frac{p}{\sigma}\right )^2 }, & \text{for } p \in \R, \\
	&h_{21}(q) = W e^{-\frac{1}{2} \left ( \frac{q}{\sigma}\right )^2 }, & \text{for } q \in \R,  \\
	&h_{22}(\bgamma_1,  n ) = \begin{cases} 1 & \text{if } \vert n \vert = \ell_2, \\ 0 & \text{otherwise,} \end{cases} & \text{for } \bgamma_1 \in \ell_1 \T,\quad n \in \ell_2\Z.
	\end{aligned}
\end{equation}
The corresponding Hamiltonian of the model $H_\omega = \pi_\omega(h)$, given by \eqref{eq:RepresentationFormulaBilayer} and sketched on Figure~\ref{fig:ToyModel}, features two intra-chain first-neighbor hopping models with amplitude normalized to $1$ for each chain, and an inter-chain coupling term with a Gaussian profile depending on the distance between the lattice sites. The two parameters of the model are the maximum amplitude $W$ and the characteristic length $\sigma$ of the inter-chain hopping interaction terms.

\subsection{Numerical approach}

Using this toy model, we aim to showcase the expected effects of continuously varying the lattice constants ratio $\alpha = \ell_2/\ell_1$ in a range around the periodic case of matched chains, $\ell_1 = \ell_2$ i.e. $\alpha = 1$.

\subsubsection{Periodic supercells}
We use here the well-known numerical strategy of scanning all periodic approximations with a given total number of sites $N$ in this range~\cite{Hofstadter1976,prodan2012quantum}, thus creating large periodic supercells. Allocating $p$ atoms to the bottom chain and $q$ atoms to the bottom chain so that $p+q = N$, we set
\begin{equation}\label{def:PeriodicSupercellParameters}
	\ell_1 = \sqrt{\frac{q}{p}}, \qquad \ell_2 = \sqrt{\frac{p}{q}} \qquad \text{so that} \qquad \ell_1 \ell_2 = 1 \quad \text{and} \quad \alpha = \frac{\ell_2}{\ell_1} = \frac{p}{q}.
\end{equation}
Note that the periodic supercell length is $\sqrt{pq} = p\ell_1  = q\ell_2$.
We then scan the range of ratios $1/6 \leq \alpha \leq 6$ by varying $p$ from $p_\mathrm{min} = \left \lceil \frac{N}{7} \right \rceil$ to $p_\mathrm{max} = \left \lfloor \frac{6N}{7} \right \rfloor$. The corresponding tight-binding hamiltonian matrix $H_0^\alpha$ with periodic boundary conditions is then assembled from the $C^*$-algebra element $h^\alpha$~\eqref{def:ModelHamiltonian} using a cut-off distance of $6 \sigma$ for the Gaussian inter-layer term, and choosing the configuration $\gamma_1 = \gamma_2 = 0$.

Our two quantities of interest are as follows.
\begin{itemize}
 	\item the density of states, i.e., the spectral measure $\mathrm{d}\mu^\alpha(E)$ on $\R$ with support on $\sigma(h)$ is defined by
\begin{equation}\label{def:DoS}
	\int_{\R} \phi(E) \mathrm{d}\mu^\alpha(E) = \Trace{\phi(h^\alpha)} \approx \frac{1}{N} \mathrm{Tr}\left (\phi(H_0^\alpha) \right ),
\end{equation}
where the test function $\phi$ is analytic in an open neighborhood of $\sigma(h^\alpha)$;
	\item the conductivity at zero frequency given by the Kubo formula~\eqref{eq:KuboFormula}, which can be computed more efficiently by introducing the current-current correlation measure $\mathrm{d}\mathcal{M}^\alpha$:
	\begin{equation}\label{eq:Conductivity_ccc}
		\sigma = \left ( \frac{e}{\hbar} \right )^2 \iint_{\R^2} \frac{f_{\beta, \mu}(E') - f_{\beta, \mu}(E)}{E-E'} \frac{\mathrm{d}\mathcal{M}^\alpha(E,E')}{1/\tau_\mathrm{rel} - i/\hbar(E-E') -i\widehat{\omega}}.
	\end{equation}
	The spectral measure $\mathrm{d}\mathcal{M}^\alpha$ on $\R^2$ is defined by~\cite{prodan2016mapping,schulz1998kinetic}:
	\begin{equation}\label{def:cccMeasure}
	\begin{aligned}
		\iint_{\R^2} \phi_1(E)\phi_2(E') &\mathrm{d}\mathcal{M}^\alpha(E,E') = \Trace{\phi_1(h)\cdot \partial_1 h \cdot \phi_2(h) \cdot \partial_1 h}\\
		&  \approx \frac{1}{N}  \mathrm{Tr} \left ( \phi_1(H_0^\alpha)\cdot \widetilde{\partial_1 H_0^\alpha} \cdot \phi_2(H_0^\alpha) \cdot \widetilde{\partial_1 H_0^\alpha} \right ),
	\end{aligned}
	\end{equation}
where the test functions $\phi_1, \phi_2$ are analytic in an open neighborhood of $\sigma(h^\alpha)$. Since $x \mapsto x$ is not $\sqrt{pq}$-periodic, we have used in~\eqref{def:cccMeasure} the approximate periodic differential calculus introduced by Prodan~\cite{prodan2012quantum}, such that the approximate derivation $\widetilde{\partial_1 H_0^\alpha}$ is obtained by :
\begin{equation}\label{def:ApproxDiffCalculus}
	\widetilde{\partial_1 f}(\omega, x) = i \widetilde{\mathfrak{X}}(x) f(\omega, x),
\end{equation}
where $x \mapsto \widetilde{\mathfrak{X}}(x)$ is $\sqrt{pq}$-periodic and approximates the identity $x \mapsto x$ near zero.
\end{itemize}
\begin{remark}
	Note that knowledge of the current-current correlation measure, which depends only on the Hamiltonian, is enough to compute the conductivity for any values of the Fermi level $\mu$ or Boltzmann factor $\beta = 1/k_B T$.
\end{remark}

\subsubsection{Kernel polynomial method}
 The second ingredient in our calculations is a Chebyshev polynomial expansion, commonly called the Kernel Polynomial Method~\cite{Weisse2006}. The idea is to rescale the Hamiltonian, $H_0^\alpha \to \hat{H}_0^\alpha = (H^\alpha_0-b)/a$ with $a,b$ well chosen so that the spectrum $\sigma (\hat{H}_0^\alpha )$ lies in the energy range $(-1, 1)$.

 \paragraph{Density of States.} The moments of the spectral measure $\mathrm{d}\mu^\alpha$ can be computed on the basis of Chebyshev polynomials of the first kind $\{T_m \}_{m \geq 0}$, which form an orthonormal basis of $L^2([-1,1])$ with respect to the weight function $w(x) = 1/(\pi \sqrt{1-x^2})$. These polynomials obey the recursion relation,
 \begin{equation}\label{eq:ChebRecursion}
 	\begin{aligned}
 		&T_0(x) = 1, \qquad T_1(x) = x, \\
 		&T_{m+1}(x) = 2 x T_m(x) - T_{m-1}(x) \qquad \text{for } m \geq 1.
 	\end{aligned}
 \end{equation}
 From~\eqref{def:DoS}, we thus obtain that
\begin{equation}\label{def:DoSChebMoments}
	\mu^\alpha_m = \int_{\R} T_m\left (\frac{E-b}{a} \right) \mathrm{d}\mu^\alpha(E) \approx \frac{1}{N} \mathrm{Tr}\left (T_m(\hat{H}_0^\alpha) \right ) =  \frac{1}{N} \sum_{j=1}^N T_m(\hat{\lambda}_j^\alpha),
\end{equation}
where $\{\hat{\lambda}_j^\alpha \}_{1 \leq j \leq N}$ is the vector of eigenvalues of $\hat{H}_0^\alpha$. Utilizing the recursion~\eqref{eq:ChebRecursion}, the moments $(\mu^\alpha_m)_{0 \leq m \leq M}$ can be computed efficiently up to some fixed polynomial degree $M$.
Now, assuming that $h^\alpha$ has an absolutely continuous spectrum, the density of states is continuous with respect to the Lebesgue measure, $\mathrm{d}\mu^\alpha(E) = \nu^\alpha(E) dE$, we can reconstruct accurately the spectral density $\nu^\alpha$ from the Chebyshev moments~\cite{Weisse2006}:
\begin{equation}\label{eq:ChebDoSReconstruction}
	\nu^\alpha(E) = \frac{1}{\pi\sqrt{a^2 - ( E-b )^2}} \left ( \mu_0 + 2 \sum_{m = 1}^M \mu^\alpha_m g^M_m T_m\left ( \frac{E-b}{a} \right) \right ),
\end{equation}
where $g^M_m = [ (M-m+1) \cos(\frac{\pi m}{M+1}) + \sin(\frac{\pi m}{M+1}) \cot(\frac{\pi}{M+1})] / (M+1)$ are the Jackson damping coefficients designed to avoid spurious Gibbs oscillations.

Finally, we note that the values of $\nu^\alpha$ at the particular set of points
\begin{equation}\label{def:ChebGaussPoints}
	x_k = a \cos \left ( \frac{\pi (k+1/2)}{M}\right) + b \qquad \text{with } k = 0, \dots, M-1,
\end{equation}
coinciding with the abscissas of the Chebyshev-Gauss numerical integration points, can be obtained through a fast cosine transform in $\mathcal{O}(M \log M)$ operations since
\[
	\gamma^\alpha_k = \pi\sqrt{a^2 - ( x_k-b )^2}\nu^\alpha(x_k) = \sum_{m = 0}^M (1+ \delta_{m,0})\mu^\alpha_m g^M_m \cos \left ( \frac{\pi m (k+1/2)}{M}\right ).
\]
This further diminishes the numerical cost of evaluating~\eqref{eq:ChebDoSReconstruction}.

\paragraph{Conductivity.} Let us define the two-dimensional Chebyshev moments of the current-current correlation measure $\mathrm{d}\mathcal{M}^\alpha$ following \eqref{def:cccMeasure}:
\begin{equation*}
	\begin{aligned}
		\mathcal{M}^\alpha_{mn} = \iint_{\R^2} T_m\left (\frac{E-b}{a} \right)&T_n\left (\frac{E'-b}{a} \right) \mathrm{d}\mathcal{M}^\alpha(E,E') \\
		\approx & \frac{1}{N}  \mathrm{Tr} \left ( T_m(\hat{H}_0^\alpha)\cdot \widetilde{\partial_1 H_0^\alpha} \cdot T_n(\hat{H}_0^\alpha) \cdot \widetilde{\partial_1 H_0^\alpha} \right ).
	\end{aligned}	
\end{equation*}
To simplify the computation, let us introduce the diagonal matrix of eigenvalues $\hat{D}^\alpha = \mathrm{Diag}(\lambda^\alpha_1, \dots, \lambda^\alpha_N)$, a unitary matrix of eigenvectors $V^\alpha$ of $\hat{H}_0^\alpha,$ and the Hermitian matrix $J^\alpha$ such that
\[
	J^\alpha = (V^\alpha)^*  \widetilde{\partial_1 H_0^\alpha} V^\alpha \hspace{1cm} \text{and} \hspace{1cm} \hat{H}_0^\alpha = V^\alpha \hat{D}^\alpha (V^\alpha)^*.
\]
Then,
\begin{equation}\label{def:cccMeasureChebMoments}
	\begin{aligned}
		\mathcal{M}^\alpha_{mn} \approx \frac{1}{N}  \mathrm{Tr} \left ( T_m(\hat{D}^\alpha)\cdot J^\alpha \cdot T_n(\hat{D}^\alpha) \cdot J^\alpha \right )
		 = \frac{1}{N} \sum_{i,j=1}^N T_m(\hat{\lambda}^\alpha_i)  \vert J^\alpha_{ij} \vert^2 T_n(\hat{\lambda}^\alpha_j).
	\end{aligned}	
\end{equation}
The moments $\mathcal{M}^\alpha_{mn}$ can thus be efficiently computed up to the partial degree $M$, using the recursion~\eqref{eq:ChebRecursion} as before. They can then be used to evaluate the conductivity by Chebyshev-Gauss numerical integration in~\eqref{eq:Conductivity_ccc} as follows. Denoting $\Phi_{\beta, \mu, \tau_\mathrm{rel},\widehat{\omega}}(E,E')$ the integrand in the right-hand side of~\eqref{eq:Conductivity_ccc}, we introduce approximate Chebyshev moments $\Phi_{\beta, \mu, \tau_\mathrm{rel},\widehat{\omega}}^{mn}$ computed with Chebyshev-Gauss integration on the points $x_k$ defined by~\eqref{def:ChebGaussPoints}:
\[
	\begin{aligned}
		\Phi_{\beta, \mu, \tau_\mathrm{rel},\widehat{\omega}}^{mn} &= \frac{(1+\delta_{m,0})(1+\delta_{n,0})}{M^2} \sum_{k,l=0}^{M-1} \Phi_{\beta, \mu, \tau_\mathrm{rel},\widehat{\omega}}(b+ax_k, b+ax_l) T_m(x_k)T_n(x_l) \\
		& \approx (1+\delta_{m,0})(1+\delta_{n,0}) \iint_{[-1,1]^2} \frac{\Phi_{\beta, \mu, \tau_\mathrm{rel},\widehat{\omega}}(b+ax, b+ay)T_m(x)T_n(y)}{\pi^2 \sqrt{1-x^2}\sqrt{1-y^2}} \mathrm{d}x\mathrm{d}y,
	\end{aligned}
\]
so that the function $\Phi_{\beta, \mu, \tau_\mathrm{rel}}$ is well approximated by the expansion
\[
	\Phi_{\beta, \mu, \tau_\mathrm{rel},\widehat{\omega}}(E,E') \approx \sum_{m,n = 0}^M \Phi_{\beta, \mu, \tau_\mathrm{rel},\widehat{\omega}}^{mn}  g^M_m g^M_n T_m\left (\frac{E-b}{a} \right) T_n\left (\frac{E'-b}{a} \right).
\]
Thanks to~\eqref{eq:Conductivity_ccc}, the conductivity is then given by:
\[
	\begin{aligned}
		\sigma &= \iint_{\R^2} \Phi_{\beta, \mu, \tau_\mathrm{rel},\widehat{\omega}}(E,E')\mathrm{d}\mathcal{M}^\alpha(E,E') \approx \sum_{m,n = 0}^M \Phi_{\beta, \mu, \tau_\mathrm{rel},\widehat{\omega}}^{mn} g^M_m g^M_n \mathcal{M}^\alpha_{mn} \\
			&\approx \sum_{m,n = 0}^M \frac{(1+\delta_{m,0})(1+\delta_{n,0})}{M^2} \sum_{k,l=0}^{M-1} \Phi_{\beta, \mu, \tau_\mathrm{rel},\widehat{\omega}}(b+ax_k, b+ax_l) T_m(x_k)T_n(x_l) g^M_m g^M_n \mathcal{M}^\alpha_{mn}.
	\end{aligned}
\]
Finally, we use the relations $T_m(x_k) = \cos \left ( \frac{\pi m (k+1/2)}{M}\right )$ and $T_n(x_l) = \cos \left ( \frac{\pi n (l+1/2)}{M}\right )$ and we exchange the order of summations. The numerical approximation to the conductivity is then given by the quadrature formula
\begin{equation}\label{eq:ChebConductivityQuadrature}
	\sigma \approx \frac{1}{M^2}\sum_{k,l=0}^{M-1} \Gamma^\alpha_{kl} \Phi_{\beta, \mu, \tau_\mathrm{rel},\widehat{\omega}}(ax_k + b,ax_l + b),
\end{equation}
where the $\Gamma_{kl}^\alpha$ are the raw output of the 2D fast cosine transform
\[
	\Gamma^\alpha_{kl} = \sum_{m,n = 0}^M (1+\delta_{m,0})(1+\delta_{n,0})\mathcal{M}^\alpha_{mn} g^M_m g^M_n\cos \left ( \frac{\pi m (k+1/2)}{M}\right )\cos \left ( \frac{\pi n (l+1/2)}{M}\right ).
\]
\begin{table}[b!]
\centering
	\begin{tabular}{|c|c|c|c|c|c|c|c|c|}
	\hline
	$N$ & $p_\mathrm{min}$ & $p_\mathrm{max}$ & $M$ & $W$ & $\sigma$ & $\beta$ & $\tau_\mathrm{rel}$ & $\widehat{\omega}$\\
	\hline
	$4181$ & $597$ & $3583$ & $1000$ & $.5$ & $.25$ & $250$ & $250$ & $0$\\
	\hline
	\end{tabular}
	\caption{Choice of numerical parameters}\label{table:NumParameters}
\end{table}

In conclusion, formulae~\eqref{eq:ChebDoSReconstruction} and~\eqref{eq:ChebConductivityQuadrature} show that the Chebyshev expansion leads to an efficient representation of the spectral measures, relying on fast discrete cosine transforms to compute weights $\gamma^\alpha_k$ and $\Gamma^\alpha_{kl}$ that depend only on the matrix $H^\alpha_0$. The knowledge of these coefficients enables accurate calculations of the quantities of interest for any choice of the energy $E$ for the density of states, or of the four parameters $\mu$, $\beta$, $\tau_\mathrm{rel}$ or $\widehat{\omega}$ for the conductivity.

\begin{remark}
	In this work, we present some results for relatively small Hamiltonian matrices, for which a full diagonalization is feasible and provides an effective way of computing the quantities of interest. For larger matrices, a promising alternative is the stochastic evaluation of traces which can be used to compute Chebyshev moments such as~\eqref{def:cccMeasureChebMoments} needed in the Kernel Polynomial Method, losing some accuracy but enabling much larger calculations~\cite{Weisse2006}.
\end{remark}

\subsection{Numerical results}
\begin{figure}[t!]
	\centering
	\includegraphics[width=.8\textwidth]{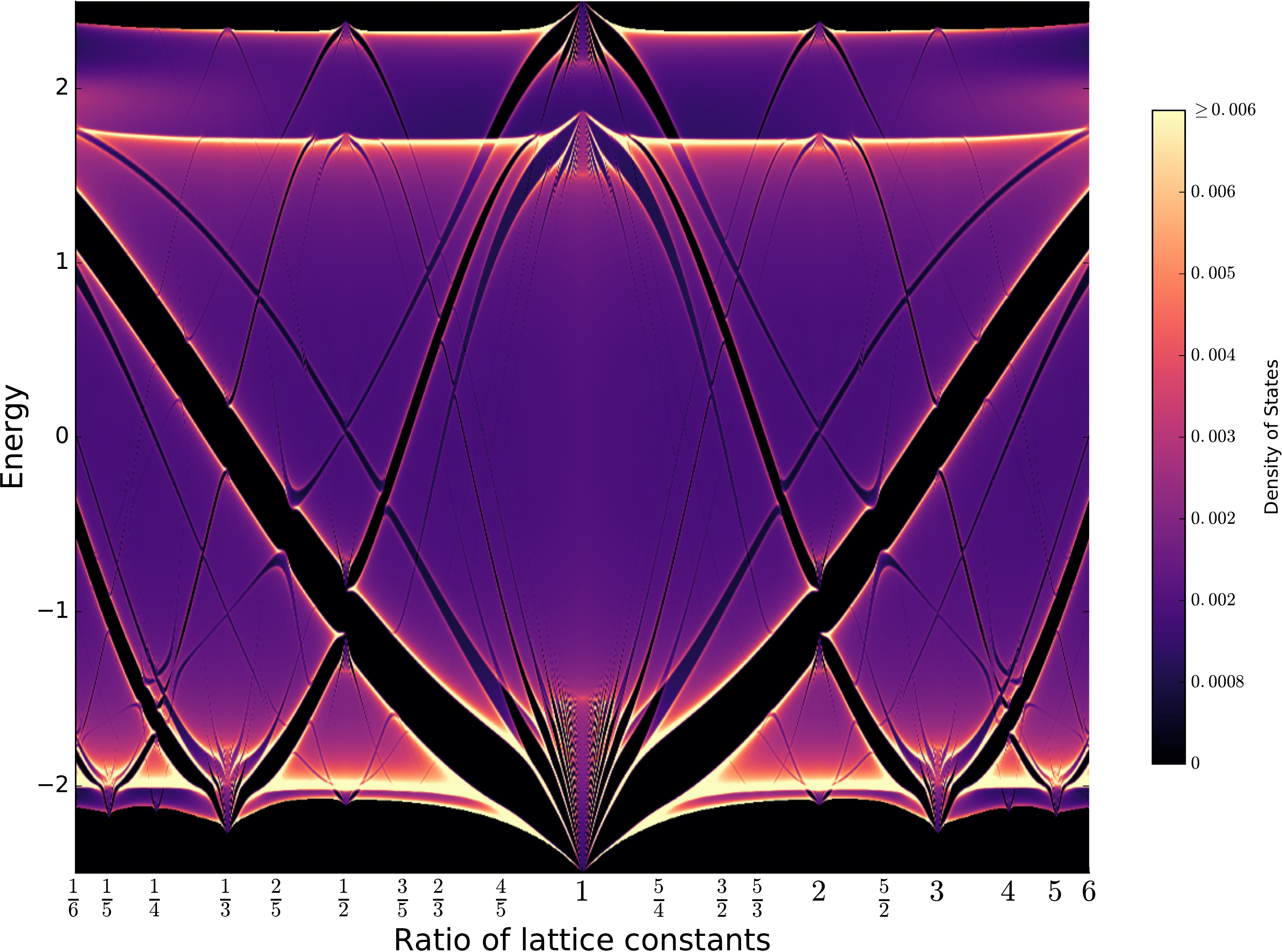}
	\caption{Density of states in color scale as a function of energy and lattice constants ratios.}
	\label{fig:DoS}
\end{figure}

For the purpose of the numerical examples presented in this section, we use the set of numerical parameters presented in Table~\ref{table:NumParameters}.
 The numerical strategy proposed above was implemented in Julia~\cite{Julia12}. We plot first the density of states as a function of lattice constants ratio $\alpha$ and energy $E$ in Figure~\ref{fig:DoS}.
 A clear fractal pattern of band gaps emerges, with continuous dependence on the lattice constants ratio parameter.
  Note the divergence of the density of states at the edges of the gaps due to the one-dimensional nature of the system.

 The overall pattern is reminiscent of the Hofstadter butterfly~\cite{Hofstadter1976}, which is a paradigm of fractal structure in the density of states of an electronic Hamiltonian induced by the interplay between two length scales (lattice and magnetic field), measured by the magnetic flux through the unit cell.
 In particular, around $\alpha = \ell_2 / \ell_1 = 1$ a large number of gaps open at the top and bottom of the spectrum. Although the resolution in this region is not very good, the similarity with Landau levels for which the energy is proportional to the magnetic field (the incommensurability parameter) is striking.

\begin{figure}[t!]
	\centering
	\includegraphics[width=.7\textwidth]{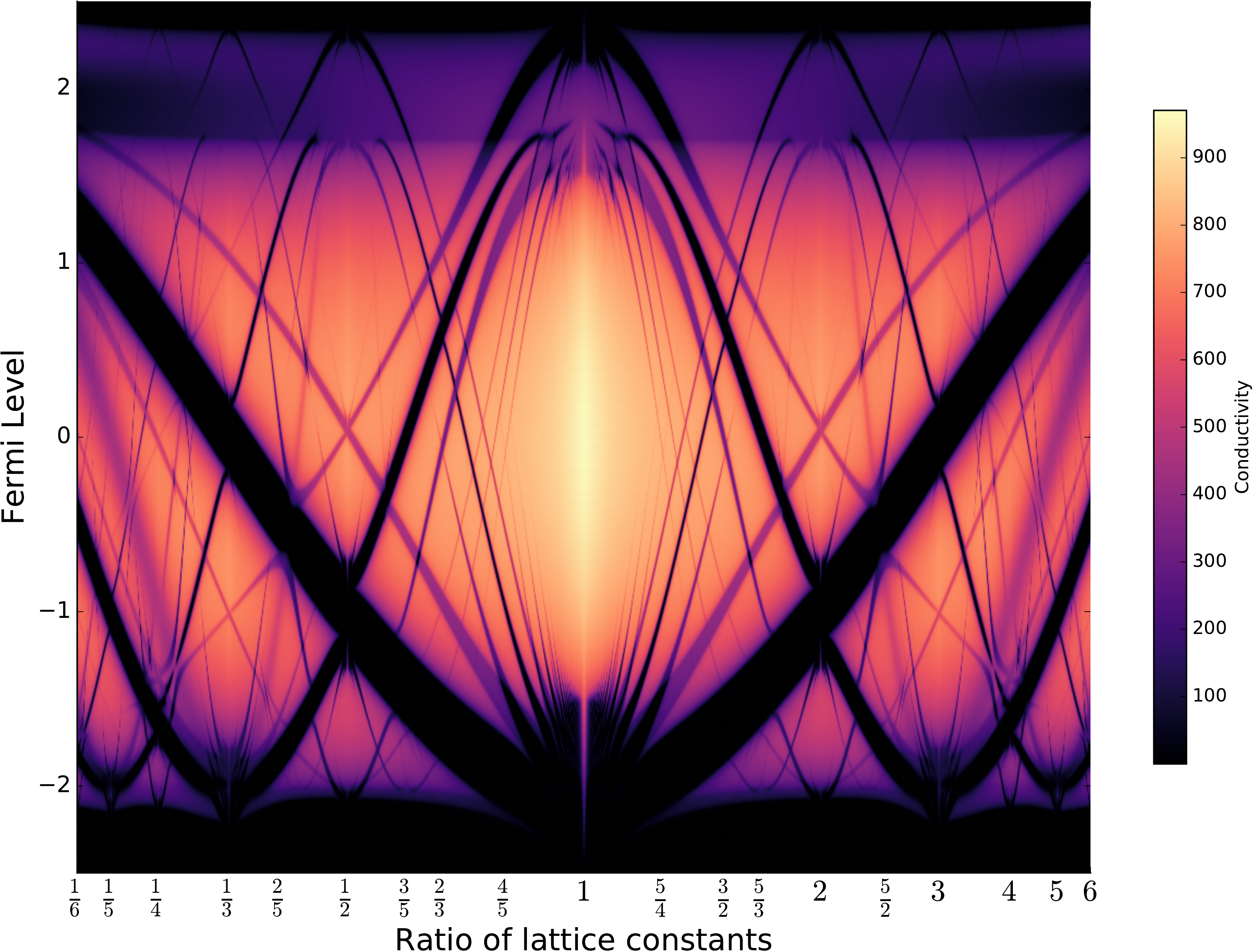}
	\caption{Conductivity in color scale as a function of Fermi level and lattice constants ratios, with parameter choice $W = .5$,  $\sigma = .25$, $\beta = \tau = 250$.}
	\label{fig:conductivity}
\end{figure}
\begin{figure}[b!]
	\centering
	\includegraphics[width=.7\textwidth]{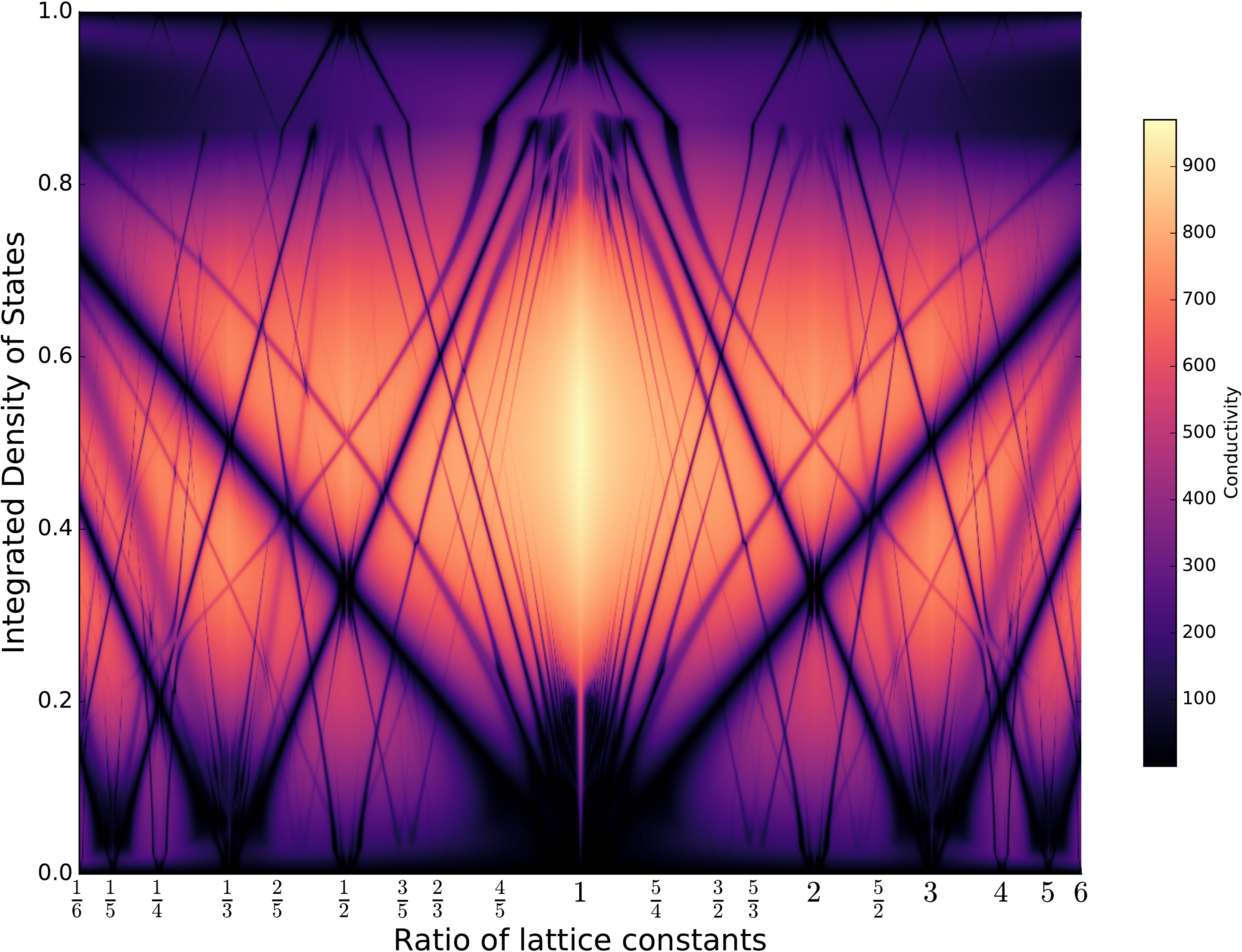}
	\caption{Conductivity in color scale as a function of the integrated density of states.}
	\label{fig:conductivity_2}
\end{figure}

Next, we plot the conductivity as a function of lattice constants ratio $\alpha$ and Fermi level $\mu$ in Figure~\ref{fig:conductivity}. The same fractal pattern emerges, however a striking difference is that the conductivity drops at the edges of the gaps (where the density of states is maximum). Strong insulating gaps occur along the spectral gaps as expected, forming a strong fan structure around $\alpha = 1$.

This fan structure is repeated at a number of values of $\alpha$ corresponding to rational numbers: $1/5$ and $5$, $1/3$ and $3$ are particularly strong, with weaker features at $1/4$ and $4$, $1/2$ and $2$, $3/5$ and $5/3$.
In general, features such as gaps opening or closing appear at rational values of $\alpha$.

Finally, we plot in Figure~\ref{fig:conductivity_2} the conductivity again, which we plot this time as a function of the integrated density of states
\[
	n^\alpha(E) = \int_{-\infty}^E \mathrm{d}\mu^\alpha,
\]
where $\mathrm{d}\mu^\alpha$ is the density of states measure defined by~\eqref{def:DoS}.
The function $E \mapsto n^\alpha(E)$ is increasing from $0$ to $1$, but stays constant in a gap of the spectrum, thus in this representation the size of the insulating gaps is narrowed. Note that in experiments, the control is over the number of electrons per unit cell through gating or doping~\cite{Sugawara2011}, corresponding to the integrated density of states, and not directly over the Fermi level in general.

Surprisingly, the gaps are still clearly visible in this rescaled presentation, but appear as straight lines. This feature is reminiscent of experimental images of magneto-transport data in small twist angle bilayer graphene~\cite{Cao2016}, leading to exciting perspectives for the application of the framework presented in this paper to more realistic 2D multilayer systems. 
\section{Acknowledgments}
This work was supported in part by ARO MURI Award W911NF-14-1-0247. EC and PC are grateful to Jean Bellissard for useful discussions.
\nocite{*}

\bibliographystyle{plain}
\bibliography{biblio}

\begin{thebibliography}{10}

\bibitem{Arveson2006}
W.~Arveson.
\newblock {\em A short course on spectral theory}, volume 209.
\newblock Springer Science \& Business Media, 2006.

\bibitem{bellissard2003coherent}
J.~Bellissard.
\newblock Coherent and dissipative transport in aperiodic solids.
\newblock In {\em Lecture Notes in Physics}, volume 597, pages 413--486.
  Springer, 2003.

\bibitem{bellissard1994noncommutative}
J.~Bellissard, A.~van Elst, and H.~Schulz-Baldes.
\newblock The noncommutative geometry of the quantum {H}all effect.
\newblock {\em Journal of Mathematical Physics}, 35(10):5373--5451, 1994.

\bibitem{Julia12}
J.~Bezanson, S.~Karpinski, V.~B. Shah, and A.~Edelman.
\newblock Julia: A fast dynamic language for technical computing.
\newblock {\em arXiv e-print arXiv:1209.5145}, 2012.

\bibitem{Cao2016}
Y.~Cao, J.~Y. Luo, V.~Fatemi, S.~Fang, J.~D. Sanchez-Yamagishi, K.~Watanabe,
  T.~Taniguchi, E.~Kaxiras, and P.~Jarillo-Herrero.
\newblock Superlattice-induced insulating states and valley-protected orbits in
  twisted bilayer graphene.
\newblock {\em Physical Review Letters}, 117(11):116804, 2016.

\bibitem{Carr2016}
S.~Carr, D.~Massatt, P.~Cazeaux, S.~Fang, M.~Luskin, and E.~Kaxiras.
\newblock Twistronics: Manipulating the electronic properties of
  two-dimensional layered structures through their twist angle.
\newblock {\em arXiv preprint arXiv:1611.00649}, 2016.

\bibitem{Castro2009}
A.~H. Castro~Neto, F.~Guinea, N.~M.~R. Peres, K.~S. Novoselov, and A.~K. Geim.
\newblock The electronic properties of graphene.
\newblock {\em Rev. Mod. Phys.}, 81:109--162, January 2009.

\bibitem{cazeauxrippling}
P.~Cazeaux, M.~Luskin, and E.~B. Tadmor.
\newblock {Analysis of rippling in incommensurate one-dimensional coupled
  chains}.
\newblock {\em to appear in Multiscale Modeling and Simulation}, June 2016.

\bibitem{dani2012ergodic}
M.~Einsiedler and T.~Ward.
\newblock {\em Ergodic Theory: with a view towards Number Theory}.
\newblock Graduate Texts in Mathematics. Springer London, 2010.

\bibitem{Fang2016}
S.~Fang and E.~Kaxiras.
\newblock {Electronic structure theory of weakly interacting bilayers}.
\newblock {\em Physical Review B}, 93(23):235153, June 2016.

\bibitem{Fang2015}
S.~Fang, R.~{Kuate Defo}, S.~N. Shirodkar, S.~Lieu, G.~A. Tritsaris, and
  E.~Kaxiras.
\newblock Ab initio tight-binding hamiltonian for transition metal
  dichalcogenides.
\newblock {\em Physical Review B}, 92(20):205108, November 2015.

\bibitem{Geim2013}
A.~K. Geim and I.~V. Grigorieva.
\newblock {V}an der {W}aals heterostructures.
\newblock {\em Nature}, 499(7459):419--25, July 2013.

\bibitem{Hofstadter1976}
D.~R. Hofstadter.
\newblock Energy levels and wave functions of bloch electrons in rational and
  irrational magnetic fields.
\newblock {\em Physical Review B}, 14:2239--2249, September 1976.

\bibitem{Kaxiras_2003}
E.~Kaxiras.
\newblock {\em Atomic and Electronic Structure of Solids}.
\newblock Cambridge University Press, Cambridge, 2003.

\bibitem{kohn1959analytic}
W.~Kohn.
\newblock Analytic properties of {B}loch waves and {W}annier functions.
\newblock {\em Physical Review}, 115(4):809, 1959.

\bibitem{marzari1997maximally}
N.~Marzari and D.~Vanderbilt.
\newblock Maximally localized generalized {W}annier functions for composite
  energy bands.
\newblock {\em Physical review B}, 56(20):12847, 1997.

\bibitem{MassattDOS16}
D.~Massatt, M.~Luskin, and C.~Ortner.
\newblock Electronic density of states for incommensurate layers.
\newblock {\em arXiv preprint arxiv:1608.01968}, August 2016.

\bibitem{prodan2012quantum}
E.~Prodan.
\newblock Quantum transport in disordered systems under magnetic fields: A
  study based on operator algebras.
\newblock {\em Applied Mathematics Research eXpress}, 2013(2):176--265, 2013.

\bibitem{prodan2016mapping}
E.~Prodan and J.~Bellissard.
\newblock Mapping the current--current correlation function near a quantum
  critical point.
\newblock {\em Annals of Physics}, 368:1--15, 2016.

\bibitem{schulz1998kinetic}
H~Schulz-Baldes and J~Bellissard.
\newblock A kinetic theory for quantum transport in aperiodic media.
\newblock {\em Journal of Statistical Physics}, 91(5-6):991--1026, 1998.

\bibitem{Sugawara2011}
K.~Sugawara, K.~Kanetani, T.~Sato, and T.~Takahashi.
\newblock Fabrication of {L}i-intercalated bilayer graphene.
\newblock {\em AIP Advances}, 1(2), 2011.

\bibitem{Terrones2014}
H.~Terrones and M.~Terrones.
\newblock Bilayers of transition metal dichalcogenides: Different stackings and
  heterostructures.
\newblock {\em Journal of Materials Research}, 29:373--382, 2 2014.

\bibitem{2DPerturb15}
G.~A. Tritsaris, S.~N. Shirodkar, E.~Kaxiras, P.~Cazeaux, M.~Luskin,
  P.~Plech\'a\v{c}, and E.~Canc\`es.
\newblock Perturbation theory for weakly coupled two-dimensional layers.
\newblock {\em Journal of Materials Research}, 31:959--966, 4 2016.

\bibitem{wannier1937structure}
G.~H Wannier.
\newblock The structure of electronic excitation levels in insulating crystals.
\newblock {\em Physical Review}, 52(3):191, 1937.

\bibitem{Weisse2006}
A.~Wei\ss{}e, G.~Wellein, A.~Alvermann, and H.~Fehske.
\newblock The kernel polynomial method.
\newblock {\em Reviews of Modern Physics}, 78:275--306, Mar 2006.

\end{thebibliography}
\end{document}